\documentclass[journal]{IEEEtran}
\usepackage[english]{babel}
\usepackage[latin1]{inputenc}
\usepackage{enumerate}
\usepackage{color}
\usepackage[T1]{fontenc}
\usepackage{subfigure}
\usepackage{dsfont}
\usepackage[T1]{fontenc}
\usepackage{amsmath}
\usepackage{amsthm}
\usepackage{amstext}
\usepackage{amssymb}
\usepackage{mathrsfs}
\usepackage{cite}
\usepackage{mathtools}
\usepackage{tikz}
\usetikzlibrary{arrows}

\newcommand{\appsec}{
\renewcommand{\thesubsection}{\Alph{subsection}}
}

\def\R{\mathbb{R}}

\def\supp{\mathop{\rm supp}}

\def\argmin{\mathop{\rm arg\, min}}

\def\A{{\mathcal A}}
\def\B{{\mathcal B}}

\def\P{{\mathcal P}}
\def\Q{{\mathcal Q}}

\def\M{{\mathcal M}}

\def\sq{{\mathsf q}}
\def\sX{{\mathsf X}}
\def\sY{{\mathsf Y}}
\def\sZ{{\mathsf Z}}
\def\sM{{\mathsf M}}
\def\sE{{\mathsf E}}

\newtheorem{theorem}{Theorem}
\newtheorem{corollary}{Corollary}

\newtheorem{proposition}{Proposition}
\newtheorem{lemma}{Lemma}
\theoremstyle{remark}
\newtheorem{remark}{Remark}
\IEEEoverridecommandlockouts

\allowdisplaybreaks 

\begin{document}
\sloppy
\title{Randomized Quantization and Source Coding  with Constrained Output Distribution
\thanks{The authors are with the Department of Mathematics and Statistics,
     Queen's University, Kingston, ON, Canada, 
     Email: \{nsaldi,linder,yuksel\}@mast.queensu.ca}
\thanks{This research was supported in part by
  the Natural Sciences and
  Engineering Research Council (NSERC) of Canada.}
\thanks{The material in this paper was presented in part at the 2013 IEEE
  International Symposium on Information Theory, Istanbul, Turkey, July 2013.}
\thanks{\copyright 2014 IEEE. Personal use
of this material is permitted.  However, permission to use this material for
any other purposes must be obtained from the IEEE by sending a request to
pubs-permissions@ieee.org.}
}
\author{Naci Saldi, Tam\'{a}s Linder, Serdar Y\"uksel}  
\maketitle
\begin{abstract}
  This paper studies fixed-rate randomized vector quantization under the
  constraint that the quantizer's output has a given fixed probability
  distribution. A general representation of randomized quantizers that includes
  the common models in the literature is introduced via appropriate mixtures of
  joint probability measures on the product of the source and reproduction
  alphabets.  Using this representation and results from optimal transport
  theory, the existence of an optimal (minimum distortion) randomized quantizer
  having a given output distribution is shown under various
  conditions. For sources with densities and the mean square distortion measure,
  it is shown that this optimum can be attained by randomizing quantizers having
  convex codecells. For stationary and memoryless source and output
  distributions a rate-distortion theorem is proved, providing a single-letter
  expression for the optimum distortion in the limit of large block-lengths.

\end{abstract}


\section{Introduction}

A quantizer maps a source (input) alphabet into a finite collection of points
(output levels) from a reproduction alphabet. A quantizer's performance is
usually characterized by its rate, defined as the logarithm of the number of
output levels, and its expected distortion when the input is a random
variable. One usually talks about randomized quantization when the quantizer
used to encode the input signal is randomly selected from a given collection of
quantizers. Although introducing randomization in the quantization procedure
does not improve the optimal rate-distortion tradeoff, randomized quantizers may
have other advantages over their deterministic counterparts. 

In what appears to be the first work explicitly dealing with randomized quantization,
Roberts \cite{Rob62} found that adding random noise to an image before
quantization and subtracting the noise before reconstruction may
result in a perceptually more pleasing image.  Schuchman \cite{Sch64}
and Gray and Stockham \cite{GrSt93} analyzed versions of such so
called \emph{dithered} scalar quantizers where random noise (dither)
is added to the input signal prior to uniform quantization. If the
dither is subtracted after the quantization operation, the procedure
is called subtractive dithering; otherwise it is called
non-subtractive dithering. Under certain conditions, dithering results
in uniformly distributed quantization noise that is independent of the
input \cite{Sch64,GrSt93}, which allows a simple modeling of the
quantization process by an additive noise channel.  In the information
theoretic literature the properties of entropy coded dithered lattice
quantizers have been extensively studied. For example,  such quantizers have
been used  to provide achievable bounds on the
performance of universal lossy compression systems by Ziv \cite{Ziv85}
and Zamir and Feder \cite{ZaFe92,ZaFe96b}.  Recently Akyol and Rose
\cite{AkRo12-2}, \cite{AkRo13},  introduced a class of randomized \emph{nonuniform}
scalar quantizers obtained via applying companding to a dithered
uniform quantizer and investigated optimality conditions for the
design of such quantizers. One should also note that the random codes used to prove
the achievability part of Shannon's rate-distortion theorem \cite{Sha59} can
also be viewed as randomized quantizers.

Dithered uniform/lattice and companding quantizers, as well as random
rate-distortion codes, pick a random quantizer from a ``small'' structured
subset of all possible quantizers. Such special randomized quantizers may be
suboptimal for certain tasks and one would like to be able to work with more
general (or completely general) classes of randomized quantizers. For example,
Li \emph{et al$.$} \cite{LiKlKl10} and Klejsa \emph{at al$.$} \cite{KlZhLiKl13}
considered \emph{distribution-preserving} 
dithered scalar quantization, where the quantizer output is restricted to have
the same distribution as the source, to improve the perceptual quality of mean
square optimal quantizers in audio and video coding. Dithered quantizers or
other structured randomized quantizers classes likely do not provide optimal
performance in this problem. In an unpublished work \cite{LiKlKl11} the same
authors considered more general distribution-preserving randomized vector
quantizers and lower bounded the minimum distortion achievable by such schemes
when the source is stationary and memoryless. 

In this paper we propose a general model which formalizes the notion of randomly
picking a quantizer from the set of \emph{all} quantizers with a given number of
output levels. Note that this set is much more complex and less structured then
say the \emph{parametric} family of all quantizers having a given number of
convex codecells.  Inspired by work in stochastic control (e.g., \cite{Bor93})
our model represents the set of all quantizers acting on a given source as a
subset of all joint probability measures on the product of the source and
reproduction alphabets. Then a randomized quantizer corresponds to a mixture of
probability measures in this subset.  The usefulness of the model is
demonstrated by rigorously setting up a generalization of the
distribution-preserving quantization problem where then the goal is to find a
randomized quantizer minimizing the distortion under the constraint that the
output has a given distribution (not necessarily that of the source).  We show
that under quite general conditions an optimal solution (i.e., an optimal
randomized quantizer) exists for this generalized problem.  We also consider a
relaxed version of the output distribution constraint where the output
distribution is only required to belong to some neighborhood (in the weak
topology) of a target distribution. For this problem we show the optimality of
randomizing among finitely many quantizers.  While for fixed quantizer dimension
we can only provide existence results, for stationary and memoryless source and
output distributions we also develop a rate-distortion theorem which identifies
the minimum distortion in the limit of large block lengths in terms of the
so-called output-constrained distortion-rate function. This last result solves a
general version of a problem that was left open in \cite{LiKlKl11}.

The rest of the paper is organized as follows. In Section~\ref{sec2} we
introduce our general model for randomized quantization and show its equivalence
to other models more common in the information theoretic literature.  In
Section~\ref{sec3} the randomized quantization problem with an output
distribution constraint is formulated and the existence of an optimal solution
is shown using optimal transport theory.  For the special but important case of
sources with densities and the mean square distortion measure, we show that this
optimum can be attained by randomizing quantizers having convex codecells.  In
Section~\ref{sec4} a relaxed version of output distribution constraint is
considered where finitely randomized quantizers are optimal.  In
Section~\ref{sec5} we present and prove a rate-distortion theorem for fixed-rate
lossy source coding with an output distribution constraint.  Many of the proofs
are quite technical and they are relegated to the Appendix.

\section{Models of Randomized Quantization}
\label{sec2}

\subsection{Notation}

In this paper $\sX$ denotes the input alphabet and $\sY$ is the reconstruction
(output) alphabet. Throughout the paper we set $\sX =\sY=\R^n$, the
$n$-dimensional Euclidean space for some $n\ge 1$, although most of the results
hold in more general settings; for example if the input and output alphabets are
Polish (complete and separable metric) spaces.  If $\sE$ is a metric space,
$\B(\sE)$ and $\P(\sE)$ will denote the Borel $\sigma$-algebra on $\sE$ and the
set of probability measures on $(\sE,\B(\sE))$, respectively.  It will be tacitly
assumed that any metric space is equipped with its Borel $\sigma$-algebra and
all probability measures on such spaces will be Borel measures. The product of
metric spaces will be equipped with the product Borel $\sigma$-algebra. Unless
otherwise specified, the term ``measurable'' will refer to Borel measurability.
We always equip $\P(\sE)$ with the Borel $\sigma$-algebra $\B(\P(\sE))$
generated by the topology of weak convergence \cite{Bil99}.

\subsection{Three models of randomized quantization}

An $M$-level quantizer ($M$ is a positive integer) from the input alphabet
$\sX$ to the reconstruction alphabet $\sY$ is a Borel measurable function $q:\sX
\to \sY$ whose range $q(\sX)=\{q(x) : x\in \sX\} $ contains \emph{at
  most} $M$ points of $\sY$. If $\Q_M$ denotes the set of all $M$-level
quantizers, then our definition implies $\Q_M\subset \Q_{M+1}$ for all
$M\geq 1$. 

In what follows we define three models of randomized quantization; two that are
commonly used in the source coding literature,  and  a third abstract model
that will nevertheless prove  very useful. 

\vspace{-0.3cm}

\subsection*{Model 1}
One general model of $M$-level randomized quantization that is often used in the
information theoretic literature is depicted in Fig.~\ref{fig1}.
\begin{figure}[h]
\centering
\tikzstyle{int}=[draw, fill=white!20, minimum size=2.5em]
\tikzstyle{init} = [pin edge={to-,thin,black}]
\begin{tikzpicture}[node distance=5cm,auto,>=latex']
    \node [int, pin={[init]above:$Z$}] (a) {Encoder};
    \node (b) [left of=a,node distance=1.6cm, coordinate] {a};
    \node [int, pin={[init]above:$Z$}] (c) [right of=a] {Decoder};
    \node [coordinate] (end) [right of=c, node distance=1.6cm]{};
    \path[->] (b) edge node {$X$} (a);
    \path[-] (a) edge node {$I\in \{1,\ldots,M\}$} (c);
    \draw[->] (c) edge node {$Y$} (end) ;
\end{tikzpicture}
\caption{Randomized source code (quantizer)}
\label{fig1}
\end{figure}
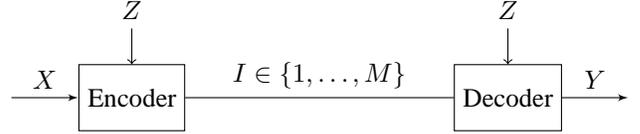

Here $X$ and $Y$ are the source and output random variables taking values in
$\sX$ and $\sY$, respectively. The index $I$ takes values in $\{1,\ldots,M\}$,
and $Z$ is a $\sZ=\R^m$-valued random variable which is independent of $X$ and
which is assumed to be available at both the encoder and the decoder. The
encoder is a measurable function $e:\sX\times\sZ\rightarrow\{1,\ldots,M\}$ which
maps $(X,Z)$ to $I$, and the decoder is a measurable function
$d:\{1,\ldots,M\}\times \sZ\rightarrow\sY$ which maps $(I,Z)$ to $Y$. For a
given source distribution, in a probabilistic sense a Model~1 quantizer is
determined by the triple $(e,d,\nu)$, where $\nu$ denotes the distribution of
$Z$.

Note that codes used in the random coding proof of the forward part of Shannon's
rate distortion theorem can be realized as Model~1 quantizers. In this case $Z$
may be taken to be the random codebook consisting of $M=2^{nR}$ code vectors of
dimension $n$, each drawn independently from a given distribution. This $Z$ can
be represented as an $m= nM$-dimensional random vector that is independent of
$X$. The encoder outputs the index $I$ of the code vector $Y$ in the codebook
that best matches $X$ (in distortion or in a joint-typicality sense) and the
decoder can reconstruct this $Y$ since it is a function of $I$ and
$Z$.

\vspace{-0.3cm}

\subsection*{Model 2}
 Model~1 can be collapsed into a more tractable equivalent model. In
 this model, a randomized quantizer is a pair $(\sq,\nu)$, where
 $\sq:\sX\times \sZ \to \sY$ is a measurable mapping such that
 $\sq(\,\cdot\,,z)$ is an $M$-level quantizer for all $z\in \sZ$ and $\nu$
 is a probability measure on $\sZ$,  the distribution of the
 randomizing random variable $Z$. Here $\sq$ is the composition of
 the encoder and the decoder in Model~1: $\sq(x,z)=d(e(x,z),z)$.

Model~2 quantizers include, as special cases, subtractive and non-subtractive
dithering of $M$-level uniform  quantizers, as well as the dithering of
non-uniform quantizers. For
example, if $n=m=1$ and $q_u$ denotes a uniform quantizer, then 
\[
\sq(x,z)=q_u(x+z)-z
\]
is a dithered uniform quantizer using subtractive dithering, 
\[
\sq(x,z)=q_u(x+z)
\]
is a dithered uniform quantizer with non-subtractive dithering, and with appropriate mappings $g$ and $h$, 
\[
\sq(x,z)=h\bigl(q_u(g(x)+z)-z\bigr).
\]
is a dithered non-uniform quantizer (e.g., \cite{LiKlKl10} and
\cite{AkRo13}). We note that dithered lattice quantizers
\cite{Ziv85,ZaFe92,Zam13} can also be considered as Model~2 type randomized
quantizers when the source has a bounded support (so that with probability one
only finitely many lattice points can occur as output points).

Let $\rho: \sX\times \sY\to \R$ be a nonnegative
measurable function, called the distortion measure. From now on we
assume that the source $X$ has distribution $\mu$ (denoted as $X\sim
\mu$). The distortion associated
with Model~2 quantizer $(\sq,\nu)$ or with Model~1 quantizer $(e,d,\nu)$, with
 $\sq(x,z)=d(e(x,z),z)$,   is the expectation
\begin{align}
\label{neweq1}
L(\sq,\nu)&=\int_{\sZ} \int_{\sX} \rho(x,\sq(x,z))\, \mu(dx) \nu(dz) \\
        &=E\bigl[\rho(X,\sq(X,Z))\bigr]  \nonumber
\end{align}
where $Z\sim \nu$ is independent of $X$.

\vspace{-0.3cm}

\subsection*{Model~3}

In this model, instead of considering quantizers as functions that map $\sX$
into a finite subset of $\sY$, first we represent them as special probability
measures on $\sX\times \sY$ (see, e.g,
\cite{BoMiSaTa05},\cite{YuLi12},\cite{Kre11},\cite{GrLu00}).  This leads to an
alternative representation where a randomized quantizer is identified with a
mixture of probability measures. In certain situations
the space of these ``mixing probabilities'' representing \emph{all} randomized
quantizers will turn out to be more tractable than considering the quite
unstructured space of all Model~1 triples $(e,d,\nu)$ or Model~2 pairs
$(\sq,\nu)$.

A \emph{stochastic kernel} \cite{HeLa96} (or \emph{regular conditional
  probability} \cite{Dud89}) on $\sY$ given $\sX$ is a function $Q(dy|x)$ such
that for each $x\in\sX$, $Q(\,\cdot\,|x)$ is a probability measure on $\sY$, and for
each Borel set $B \subset \sX$, $Q(B|\,\cdot\,)$ is a measurable function from $\sX$
to $[0,1]$. A quantizer $q$ from $\sX$ into $\sY$ can be represented as a
stochastic kernel $Q$ on $\sY$ given $\sX$ by letting \cite{YuLi12}, \cite{BoMiSaTa05},
\[
Q(dy|x)=\delta_{q(x)}(dy),
\]
where $\delta_u$
denotes the point mass at $u$: $\delta_u(A)=1$ if $u\in A$ and $\delta_u(A)=0$
if $u\notin A$ for any Borel set $A\subset \sY$.

If we fix the distribution $\mu$ of the source $X$, we can also represent $q$ by
the probability measure $v(dx\,
dy)=\mu(dx)\delta_{q(x)}(dy)$  on $\sX\times \sY$. Thus we can identify the set
$\Q_M$ of all $M$-level quantizers from $\sX$ to $\sY$ with the following subset
of $\mathcal{P}(\sX\times \sY)$:
\begin{eqnarray}
\label{eqquant}
\lefteqn{ \Gamma_{\mu}(M)} \\
&=& \bigl\{v \in \mathcal{P}(\sX\times \sY):
v(dx\,dy)\!=\!\mu(dx)\delta_{q(x)}(dy), \; q\in \Q_M \bigr\}. \nonumber
\end{eqnarray}
Note that $q \mapsto \mu(dx)\delta_{q(x)}(dy)$ maps $\Q_M$
onto $\Gamma_{\mu}(M)$, but this mapping is one-to-one  only if we  consider
equivalence classes of quantizers in $\Q_M$ that are equal $\mu$ almost
everywhere ($\mu$-a.e).

We equip $\mathcal{P}(\sX\times \sY)$ with the topology of weak convergence
(weak topology) which is metrizable with the Prokhorov metric, making
$\mathcal{P}(\sX\times \sY)$
into a Polish space \cite{Bil99}. The following lemma is proved in the Appendix~\ref{app1}.

\begin{lemma}
$\Gamma_{\mu}(M)$ is a Borel subset of  $\mathcal{P}(\sX\times \sY)$.
\label{theo1}
\end{lemma}

Now we are ready to introduce Model~3 for randomized
quantization. Let $P$ be a probability measure on
$\mathcal{P}(\sX\times \sY)$ which is supported on $\Gamma_{\mu}(M)$,
i.e., $P(\Gamma_{\mu}(M))=1$. Then $P$ induces a ``randomized quantizer'' $v_P\in
\mathcal{P}(\sX\times \sY)$ via
\[
v_P(A\times B)=\int_{\Gamma_{\mu}(M)} v(A\times B)\, P(dv)
\]
for Borel sets $A\subset \sX$ and $B\subset \sY$,
which we abbreviate to
\begin{equation}
v_P=\int_{\Gamma_{\mu}(M)} v\, P(dv).
\label{neweq2}
\end{equation}
Since each $v$ in $\Gamma_{\mu}(M)$ corresponds to a quantizer with
input distribution $\mu$, $P$ can be thought as a probability measure
on the set of all $M$-level quantizers $\Q_M$.

Let $\P_{0}(\Gamma_{\mu}(M))$ denote the set of probability measures on
$\P(\sX\times \sY)$ supported on $\Gamma_{\mu}(M)$. 
We define
the set of $M$-level Model~3 randomized quantizers as
\begin{eqnarray}
\label{neweq4} 
\lefteqn{\Gamma_{\mu}^{\rm R}(M)} \\
\mbox{} \!\!\!\!&=&\!\!\!\Bigl\{v_P\in \mathcal{P}(\sX\times \sY):
v_P=\int_{\Gamma_{\mu}(M)}\!\!\!\!\!\!\!\!v P(dv), \;P\in
\P_{0}(\Gamma_{\mu}(M))\Bigr\}.  \nonumber
\end{eqnarray}

Note that if $v_P\in \Gamma_{\mu}^{\rm R}(M)$ is a Model~3 quantizer, then the
$\sX$-marginal of $v_P$ is equal to $\mu$, and if $X$ and $Y$ are  random vectors
(defined on the same probability space)  with joint distribution $v_P$, then they 
provide a stochastic representation of the random quantizer's input and output,
respectively. Furthermore, the distortion associated with $v_P$ is 
\begin{align*}
 L(v_P) &\coloneqq    \int_{\sX\times \sY} \rho(x,y)v_P(dx\,dy) \\*
  &= \int_{\Gamma_{\mu}(M)}\int_{\sX\times \sY} \rho(x,y)
v(dx\,dy)P(dv) \\
&=  E\bigl[\rho(X,Y)\bigr] .
\end{align*}

\subsection{Equivalence of models}

Here we show that the three models of randomized quantization are essentially
equivalent. As before, we assume that the source distribution $\mu$ is fixed.
The following two results are proved in Appendix~\ref{app2} and
Appendix~\ref{app2a}, respectively.

\begin{theorem}
\label{newtheo1}
For each Model~2 randomized quantizer $(\sq,\nu)$
there exists a Model~3 randomized quantizer $v_P\in \Gamma_{\mu}^{\rm R}(M$)
such that $(X,Y)=(X,\sq(X,Z))$ has distribution $v_P$. Conversely, for
any  $v_P\in \Gamma_{\mu}^{\rm R}(M)$ there exists a Model~2
randomized quantizer $(\sq,\nu)$ such that $(X,\sq(X,Z))\sim v_P$.
\end{theorem}

\begin{theorem} \label{newtheo2} 
Models~1 and 2 of randomized quantization are
  equivalent in the sense of Theorem~\ref{newtheo1}.
\end{theorem}

\begin{remark} 
\label{rem1}
\begin{itemize}

\item[]
\item[(a)]  Clearly, any two equivalent randomized quantizers  have
  the same distortion. The main result of this section is
  Theorem~\ref{newtheo1}. Theorem~\ref{newtheo2} is intuitively obvious, but proving
  that any Model~2 quantizer can be decomposed into an equivalent Model~1
  quantizer with measurable encoder and decoder is not quite trivial.

\item[(b)] Since the dimension $m$ of the randomizing random vector $Z$ was
  arbitrary, we can take $m=1$ in Theorem~\ref{newtheo1}. In fact, the proof
  also implies that  any Model~2
  or 3 randomized quantizer is equivalent (in the sense of
  Theorem~\ref{newtheo1}) to a Model~2 quantizer $(\sq,\nu)$, where $\sq:\sX\times [0,1]\to
  \sY$ and $\nu$ is the uniform  distribution on $[0,1]$.

\item[(c)] Assume that $(\sZ,\A,\nu)$ is an \emph{arbitrary} probability space. For
  any randomized quantizer $\sq:\sX\times \sZ\to \sY$ in the form $\sq(X,Z)$,
  where $Z\sim \nu$ is independent of $X$, there exists a Model~3 randomized
  quantizer $v_P$ such that $(X,\sq(X,Z))\sim v_P$. This can be proved by using
  the same proof method as in Theorem \ref{newtheo1}.  In view
  of the previous remark and Theorem~\ref{newtheo1}, this means that uniform
  randomization over the unit interval $[0,1]$ suffices under the most general
  circumstances.

\item[(d)]  All results in this section remain valid if the input and
reproduction alphabets $\sX$ and $\sY$ are arbitrary uncountable
Polish spaces. In this case, uniform randomization over the unit interval
 still provides the most general model possible.
 \end{itemize}
\end{remark}

In the next two sections, Model~3 will be used to represent randomized
quantizers because it is particularly suited to treating the optimal randomized quantization problem
under an output distribution constraint.

\section{Optimal Randomized Quantization with Fixed Output Distribution}
\label{sec3}

Let $\psi$ be a probability measure on $\sY$ and let $\Lambda(M,\psi)$ denote
the set of all $M$-level Model~2 randomized quantizers $(\sq,\nu)$ such that the
output $\sq(X,Z)$ has distribution $\psi$. As before, we assume that $X\sim
\mu$, $Z\sim \nu$, and $Z$ and $X$ are independent. We want to show the
existence of a minimum-distortion randomized quantizer having output
distribution $\psi$, i.e, the existence of $(\sq^*,\nu^*)\in \Lambda(M,\psi)$
such that
\[
 L(\sq^*,\nu^*)= \inf_{(\sq,\nu)\in
\Lambda(M,\psi)} L(\sq,\nu).
\]
If we set $\psi=\mu$, the above problem is reduced to showing the existence of a
distribution-preserving randomized quantizer \cite{LiKlKl10,LiKlKl11} having
minimum distortion.

The set of $M$-level randomized quantizers is a fairly general
(nonparametric) set of functions and it seems difficult to investigate
the existence of an optimum directly.  On the other hand, Model~3
provides a tractable framework for establishing the existence of an
optimal randomized quantizer under quite general conditions.

Let $\Gamma_{\mu,\psi}$ be the set of all joint distributions
$v\in P(\sX\times \sY)$ having  $\sX$-marginal  $\mu$ and $\sY$-marginal $\psi$.
Then
\begin{equation}
  \label{eq_rqdef}
\Gamma_{\mu,\psi}^{\rm R}(M) =\Gamma_{\mu}^{\rm R}(M)\cap \Gamma_{\mu,\psi}  
\end{equation}
is the subset of Model~3 randomized quantizers which corresponds to the class of
output-distribution-constrained 
Model~2 randomized quantizers $\Lambda(M,\psi)$.

For any $v\in \P(\sX\times \sY)$ let 
\[
L(v)=\int_{\sX\times \sY} \rho(x,y) v(dx\,dy).
\]
Using these definitions, finding optimal  randomized quantizers with
a given output distribution can be  posed as
finding $v$ in $\Gamma_{\mu,\psi}^{\rm R}(M)$ which minimizes $L(v)$, i.e.,
\begin{align}
\textbf{(P1)} \text{                         }&\text{minimize } L(v)
\nonumber \\*
&\text{subject to  } v \in \Gamma_{\mu,\psi}^{\rm R}(M). \nonumber
\end{align}

We can prove the existence of the minimizer for \textbf{(P1)}  under either of the
following assumptions. Here $\|x\|$ denotes the  Euclidean norm of
$x\in\R^n$.
\smallskip

\noindent \textsc{Assumption 1}:\   $\rho(x,y)$ is continuous and
$\psi(B)=1$ for some compact subset $B$ of $\sY$.

\smallskip
\noindent \textsc{Assumption 2}:\ $\rho(x,y)=\|x-y\|^{2}$.

\begin{theorem}
\label{thmopt}
 Suppose $\inf_{ v \in \Gamma_{\mu,\psi}^{\rm R}(M)}
  L(v)<\infty$. Then there exists a minimizer  with finite cost for problem
  \emph{\textbf{(P1)}}  under either Assumption~1 or Assumption~2.
\end{theorem}

The theorem is proved in Appendix~\ref{app3} with the aid of optimal transport
theory \cite{Vil09}. The optimal transport problem
for marginals $\pi\in\P(\sX)$, $\lambda\in\P(\sY)$ and cost function
$c:\sX\times\sY\rightarrow[0,\infty]$ is defined as
\[
  \begin{split}
&\text{minimize } \int_{\sX\times\sY} c(x,y) v(dx\,dy)  \\
&\text{subject to  } v \in \Gamma_{\pi,\lambda}. 
\end{split}
\]

In the proof of Theorem~\ref{thmopt} we set up a relaxed version of the
optimization problem (\textbf{P1}). We show that if the relaxed problem has a
minimizer, then (\textbf{P1}) also has a minimizer, and then prove the existence
of a minimizer for the relaxed problem using results from optimal transport theory.

\begin{remark}
Note that the product distribution $\mu\otimes \psi$ corresponds to a
1-level randomized quantizer (the equivalent Model~2 randomized quantizer
is given  by  $\sq(x,z)=z$ and $Z\sim\nu$). Hence $ \mu\otimes \psi\in
\Gamma_{\mu,\psi}^{\rm R}(M)$ for all $M\ge 1$, and if $L(\mu\otimes
\psi)<\infty$, then the condition $\inf_{ v \in
  \Gamma_{\mu,\psi}^{\rm R}(M)} L(v)<\infty$ holds. In particular, if both $\mu$ and $\psi$ have finite second moments $\int \|x\|^{2}\mu(dx)<\infty$ and
$\int \|y\|^{2}\psi(dy)<\infty$, and $\rho(x,y)=\|x-y\|^2$ (Assumption~2), then
$\inf_{ v \in
  \Gamma_{\mu,\psi}^{\rm R}(M)} L(v)<\infty$.
\end{remark}

Optimal transport theory can also be used to
show  that, under some regularity conditions on the input
distribution and the distortion measure, the randomization can be
restricted to quantizers having a certain structure. Here we consider
 sources with densities and the mean square distortion. A
quantizer $q:\sX\rightarrow\sY$ with  output points $q(\sX)=\{y_1,\ldots,y_k\}\subset
\sY$  is said to have \emph{convex codecells} if
$q^{-1}(y_i)=\{x: q(x)=y_i\}$ is a convex subset of $\sX=\R^n$ for all
$i=1,\ldots,k$. Let $\Q_{M,{\rm c}}$ denote the set of all $M$-level
quantizers having convex codecells. The proof of the following theorem is given
in Appendix~\ref{app3a}.

\begin{theorem}
\label{thconv}
 Suppose $\rho(x,y)=\|x-y\|^2$ and $\mu$ admits a probability density
 function. Then an  optimal randomized quantizer in Theorem~\ref{thmopt}
can be obtained by randomizing over quantizers with convex cells. That is
\[
\min_{v\in\Gamma_{\mu,\psi}^{\rm R}(M)} L(v) = \min_{v\in\Gamma_{\mu,\psi}^{{\rm {R,c}}}(M)} L(v),
\]
where $\Gamma_{\mu,\psi}^{{\rm {R,c}}}(M)$ represents the Model~3 quantizers
with output distribution $\psi$ that are obtained by replacing $\Q_M$
with $\Q_{M,{\rm c}}$ in \eqref{eqquant}.
\end{theorem}

\begin{remark}
  Each quantizer having $M$ convex codecells can be described using $nM+
  (n+1)M(M-1)/2$ real parameters if $\mu$ has a density and any two quantizers that
  are $\mu$-a.e.\ equal are considered equivalent. One obtains such a parametric
  description by specifying the $M$ output points using $nM$ real parameters, and
  specifying 
  the $M$ convex polytopal codecells by $M(M-1)/2$ hyperplanes
  separating pairs of distinct codecells using $(n+1)M(M-1)/2$ real
  parameters.  Thus Theorem~\ref{thconv} replaces the
  \emph{nonparametric} family of quantizers $Q_{M}$ in Theorem~\ref{thmopt} with
  the \emph{parametric} family $Q_{M,{\rm c}}$.
\end{remark}

\section{Approximation with Finite Randomization}
\label{sec4}

Since randomized quantizers require common randomness that must be shared
between the encoder and the decoder, it is of interest to see how one can
approximate the optimal cost by randomizing over finitely many quantizers.
Clearly, if the target probability measure $\psi$ on $\sY$ is not finitely
supported, then no finite randomization exists with this output distribution. In
this section we relax the fixed output distribution constraint and consider the
problem where the output distribution belongs to some neighborhood (in the weak
topology) of $\psi$. We show that one can always find a finitely randomized
quantizer which is optimal (resp., $\varepsilon$-optimal) for this relaxed
problem if the distortion measure is continuous and bounded (resp., arbitrary).

Let $B(\psi,\delta)$ denote the open ball in $\P(\sY)$, with respect to the
Prokhorov metric \cite{Bil99}  (see also \eqref{eq:prokh} in
Appendix~\ref{app4}), having radius $\delta>0$ and centered at the target input
distribution $\psi$. Also, let $\M_{\mu,\psi}^{\delta}$ denote the set of all
$v\in \Gamma_{\mu}^{\rm R}(M)$ whose $\sY$ marginal belongs to $B(\psi,\delta)$.
That is, $\M_{\mu,\psi}^{\delta}$ represents all randomized quantizers in
$\Gamma_{\mu}^{\rm R}(M)$ whose output distribution is within distance $\delta$
of the target distribution $\psi$.  We consider the following relaxed version of
the minimization problem \textbf{(P1)}:
\begin{align}
\text{   \textbf{(P3)}             }&\text{minimize } L(v) \nonumber \\*
&\text{subject to  } v \in \M^{\delta}_{\mu,\psi}. \nonumber
\end{align}

The set of \emph{finitely randomized} quantizers in
$\Gamma_{\mu}^{\rm R}(M)$  is obtained by taking finite  mixtures of quantizers in
$\Gamma_{\mu}(M)$, i.e.,
\begin{eqnarray*}
\lefteqn{\Gamma_{\mu}^{\rm FR}(M) }\\
&=&\!\!\!\Bigl\{v_P\in \Gamma_{\mu}^{\rm R}(M):
v_P=\int_{\Gamma_{\mu}(M)}\!\!\!\!\!v P(dv), \;|\supp(P)|<\infty \Bigr\}.
\end{eqnarray*}

\begin{theorem}
Assume the distortion measure $\rho$ is continuous and bounded and let
$v\in \M^{\delta}_{\mu,\psi}$ be arbitrary. Then
there exists $v_{F}$ in $\M_{\mu,\psi}^{\delta}
\cap\Gamma_{\mu}^{{\rm FR}}(M)$ such that $L(v_{F})\leq L(v)$.
\label{thm_finite}
\end{theorem}

The proof is given in Appendix \ref{app4}.

Although the minimum in $\textbf{(P3)}$ may not be achieved by any   $v
\in \M^{\delta}_{\mu,\psi}$, the theorem implies that if the
problem has a solution, it also has a solution
in the set of finitely randomized quantizers.

\begin{corollary}
 \label{cor_finite}
Assume $\rho$ is continuous and  bounded and  suppose  there exists
$v^*\in  \M^{\delta}_{\mu,\psi}$ with $L(v^*)=\inf_{v\in
  \M^{\delta}_{\mu,\psi}}  L(v)$. Then there exists $v_F\in \M^{\delta}_{\mu,\psi}\cap
  \Gamma_{\mu}^{{\rm FR}}(M)$ such that $L(v_F)= L(v^*)$.
\end{corollary}

The continuity of $L$, implied by the boundedness and continuity of $\rho$ is
crucial in the proof of Theorem~\ref{thm_finite} and thus for
Corollary~\ref{cor_finite}.  However, the next theorem shows that for an
arbitrary $\rho$,  any $\varepsilon>0$, and $v\in
\M^{\delta}_{\mu,\psi}$, there exists $v_{F}$ in $\M_{\mu,\psi}^{\delta}
\cap\Gamma_{\mu}^{{\rm FR}}(M)$ such that $L(v_{F})\leq L(v)+\varepsilon$.  That is,
for any $\varepsilon>0$ there exists an $\varepsilon$-optimal finitely
randomized quantizer for \textbf{(P3)}. The theorem is proved in Appendix~\ref{app4a}

\begin{theorem}
\label{thm_finarb}
Let $\rho$ be an arbitrary distortion measure and assume
$\inf_{v\in \M^{\delta}_{\mu,\psi}} L(v)<\infty$. Then,
\[
\inf_{v\in
    \M^{\delta}_{\mu,\psi}\cap
  \Gamma_{\mu}^{{\rm FR}}(M)}  L(v)= \inf_{v\in
    \M^{\delta}_{\mu,\psi}}  L(v).
\]
\end{theorem}

\begin{remark}
The above results on finite randomization heavily depend on our use of the
Prokhorov metric as a measure of ``distance'' between two probability
measures. In particular, if one considers other  measures of closeness, such as the
Kullback-Leibler (KL) divergence or the total variation distance, then finite
randomization may not suffice if the target output distribution is not
discrete. In particular, if the target output distribution $\psi$ has a density
and $\tilde{\psi}$ denotes the (necessarily discrete) output distribution of any
finitely randomized quantizer, then $\tilde{\psi}$ is not absolutely continuous
with respect to $\psi$ and  for the KL divergence we have
$D_{KL}(\tilde{\psi}\|\psi)=\infty$, while for the total variation distance we 
have $\|\tilde{\psi}-\psi\|_{TV}=1$.

\end{remark}

\section{A Source Coding Theorem} 
\label{sec5} 

After proving the existence of an optimum randomized quantizer in problem
\textbf{(P1)} in Section~\ref{sec4}, one would also like to  evaluate  the minimum
distortion
\begin{equation}
  \label{eq_optdef}
L^* \coloneqq \min\{ L(v) :  v \in \Gamma_{\mu,\psi}^{\rm R}(M)\}
\end{equation}
achievable for  fixed source and output distributions $\mu$ and $\psi$ and given
number of quantization levels $M$. For any given blocklength $n$ this seems to
be a very hard problem in general. However, we are able to prove a
rate-distortion type result that explicitly identifies $L^*$ in the limit of
large block lengths $n$  if the source  and output distributions correspond to two stationary and memoryless (i.e., i.i.d$.$) processes.

With a slight abuse of the notation used in previous sections, we let
$\sX=\sY=\R$ and consider a sequence of problems \textbf{(P1)} with input and
output alphabets $\sX^n=\sY^n=\R^n$, $n\ge 1$, and corresponding source and
output distributions $\mu^n=\mu \otimes \cdots \otimes \mu $ and $\psi^n=\psi
\otimes \cdots \otimes \psi $, the $n$-fold products of a two fixed probability
measures $\mu,\psi\in \P(\R)$. To emphasize the changing block length,
$x^n=(x_1,\ldots,x_n)$ and $y^n=(y_1,\ldots,y_n)$ will denote generic elements
of $\sX^n$ and $\sY^n$, respectively.

\smallskip
\noindent \textsc{Assumption 3}:\ The distortion measure is the average squared
error given by
\[
\rho_n(x^n,y^n) =\frac{1}{n} \sum_{i=1}^{n} \rho(x_i,y_i)
\]
with $\rho(x,y) =(x-y)^2$. We assume that $\mu$ and
$\psi$ have finite second moments, i.e.,   $\int x^2
\mu(dx)<\infty$,  $\int y^2
\psi(dy)<\infty$.

\medskip

For $R \ge 0$ let $\Gamma^{\rm R}_{\mu^n,\psi^n}(2^{nR})$ denote the set of
$n$-dimensional Model~3 randomized quantizers defined in \eqref{eq_rqdef} having
input distribution $\mu^n$, output distribution $\psi^n$, and at most $2^{nR}$
levels (i.e., rate $R$). Then
\begin{equation}
  \label{eq_opdef1}
  L_n(\mu,\psi,R) \coloneqq \inf\bigl\{ L(v) :  v \in \Gamma^{\rm R}_{\mu^n,\psi^n}(2^{nR})
  \bigr\} \nonumber 
\end{equation}
is the minimum distortion achievable by such quantizers. 

We also define 
\begin{eqnarray*}
\lefteqn{ D(\mu,\psi,R) } \\
&=& \inf\bigl\{ E[\rho(X,Y)]: X\sim \mu,\; Y\sim \psi,\; I(X;Y)\le R\bigr\},  
\end{eqnarray*}
where the infimum is taken over pairs of all joint distributions of real random
variables $X$ and $Y$ such that $X$ has distribution $\mu$, $Y$ has distribution
$\psi$, and their mutual information $I(X;Y)$ is upper bounded by $R$. 

One can trivially adapt the standard proof from rate-distortion theory to show
that similar to the distortion-rate function, $D(\mu,\psi,R)$ is a convex and
nonincreasing function of $R$. Note that $D(\mu,\psi,R)$ is finite for all $R\ge
0$ by the assumption that $\mu$ and $\psi$ have finite second moments. The
distortion-rate function $D(\mu,R)$ of the i.i.d.\ source $\mu$, is obtained
from $D(\mu,\psi,R)$ as
\[
D(\mu,R)= \inf_{\psi \in \P(\sY)} D(\mu,\psi,R). 
\]

By a standard argument one can easily show that the sequence $\{n
L_n(\mu,\psi,R)\}_{n\ge 1}$ is subadditive and so $\inf_{n\ge 1} L_n(\mu,\psi,R)
= \lim_{n\to \infty} L_n(\mu,\psi,R)$. Thus the limit represents the minimum
distortion achievable with rate-$R$ randomized quantizers for an i.i.d.\ source
with marginal $\mu$ under the constraint that the output is i.i.d.\ with
marginal $\psi$. The next result proves that this limit is equal to
$D(\mu,\psi,R)$, which one could thus call the ``output-constrained
distortion-rate function.'' 

\begin{theorem}
  \label{thm_rd}
We have
\begin{equation}
  \label{eq_thmrd}
  \lim_{n\to \infty }  L_n(\mu,\psi,R) = D(\mu,\psi,R).
\end{equation}
\end{theorem}

\begin{remark} 
\label{rem3}
\begin{itemize}
\item[]
\item[(a)] As usual, the proof of the theorem consists of a converse and an
  achievability part. The converse (Lemma~\ref{alemma3} below) directly follows
  from the usual proof of the converse part of the rate-distortion theorem. In
  fact, this was first noticed in \cite{LiKlKl11} where the special case
  $\psi=\mu$ was considered and (in a different formulation) it was shown that
  for all $n$
\[
 L_n(\mu,\mu,R) \ge  D(\mu,\mu,R). 
\]
Virtually the same argument implies that $L_n(\mu,\psi,R)\ge D(\mu,\psi,R)$ for
all $n$ and $\psi$. Nevertheless, we write out the proof in Appendix~\ref{app5}
since, strictly speaking, the proof in \cite{LiKlKl11} is only valid if $\psi$
is discrete with finite (Shannon) entropy or it has a density and finite
differential entropy.

\item[(b)] The proof of the converse part (Lemma~\ref{alemma3}) is valid for
any randomized quantizer whose output $Y^n$ satisfies  $Y_i\sim \psi$,
$i=1,\ldots,n$. Thus the theorem also holds if in the definition of
$L_n(\mu,\psi,D)$, the randomized quantizers are required to have  outputs  with
identically
distributed (but not 
necessarily independent) components having  common distribution $\psi$.

\item[(c)] In \cite{LiKlKl11} it was left as an open problem if $D(\mu,\mu,R) $
  can be asymptotically achieved by a sequence of distribution-preserving
  randomized quantizers. The authors presented an incomplete achievability proof
  for the special case of Gaussian $\mu$ using dithered lattice quantization.
  We prove the achievability of $D(\mu,\psi,R)$ for arbitrary $\mu$ and $\psi$
  using a fundamentally different (but essentially non-constructive)
  approach. In particular, our proof is based on random coding where the codewords are
  uniformly distributed on the type class of an $n$-type that well approximates
  the target distribution $\psi$, combined with optimal coupling from mass
  transport theory.

\item[(d)] With only minor changes in the proof, the theorem remains valid if
  $\sX=\sY$ are arbitrary Polish spaces with metric $d$ and
  $\rho(x,y)=d(x,y)^p$ for some $p\ge 1$. In this case the finite second moment
  conditions translate into $\int d(x,x_0)^p \,\mu(dx)<\infty$ and $\int d(y,y_0)^p\,
  \psi(dy)<\infty$ for some (and thus all) $x_0,y_0\in \sX$.

\end{itemize}
\end{remark}

\noindent\emph{Proof of Theorem~\ref{thm_rd}.} $\,$ 
In this proof we use Model 2 of randomized quantization which  is
more suitable here than Model~3. Also, it is easier to deal with the rate-distortion
performance that with the distortion-rate performance. Thus, following the notation
in \cite{ZaRo01}, for $D\ge 0$ we define the \emph{minimum mutual information
  with constraint output $\psi$} as
\begin{eqnarray*}
\lefteqn{I_{m}(\mu \| \psi, D)} \\
&=& \inf \bigl\{ I(X;Y): X\sim\mu, Y\sim\psi,
E[\rho(X,Y)]\leq D\bigr\}, \nonumber
\end{eqnarray*}
where the infimum is taken over pairs of all joint distributions of $X$ with
marginal $\mu$ and $Y$ with marginal $\psi$ such that $E[\rho(X,Y)]\le D$.  If
this set of joint distributions is empty, we let $I_{m}(\mu \| \psi, D)=\infty$.
Clearly, the extended real valued functions $I_{m}(\mu \| \psi,\,\cdot\,)$ and
$D(R,\mu,\,\cdot\,)$ are inverses of each other. Hence $I_m(\mu\|\psi,\,D)$ is a nonincreasing, convex function of $D$.

The converse part of the theorem, i.e., the statement $L_{n}(\mu,\psi,R)\geq
D(R,\mu,\psi)$ for all $n\geq 1$, is directly implied by the following lemma.
The proofs of all lemmas in this section are given in Appendix~\ref{app5}.

\begin{lemma}
\label{alemma3}
For all $n\ge 1$ if a randomized quantizer has input distribution $\mu^n$, output
distribution $\psi^n$, and distortion $D$, then its rate is lower
bounded as
\[
R\ge I_{m}(\mu
\| \psi, D).
\]
\end{lemma}

In the rest of the proof we show the achievability of $D(R,\mu,\psi)$. We first
prove this for finite alphabets and then generalize to continuous alphabets.

Let $\sX=\sY$ be finite sets and assume that $\rho(x,y)= d(x,y)^p$, where $d$
is a metric on $\sX$ and $p > 0$. For each $n$ let $\psi_{n}$ be a closest 
$n$-type \cite[Chapter 11]{CoTh06} to $\psi$ in the $l_1$-distance  which is
absolutely continuous with respect to $\psi$, i.e., $\psi_{n}(y)=0$ whenever
$\psi(y)=0$.  Let $D$ be such that $I_m(\mu\|\psi,D)<\infty$, let $\varepsilon >0$
be arbitrary,  and set $R=I_{m}(\mu \| \psi,
D)+\varepsilon$. Assume $X^n\sim \mu^n$ for $n\geq1$. For each $n$ generate
$2^{nR}$ codewords uniformly and independently drawn from $T_{n}(\psi_n)$, the
\emph{type class} of $\psi_n$ \cite{CoTh06}, i.e., independently (of each other and of $X^n$)
generate random codewords $U^{n}(1),\ldots,U^{n}(2^{nR})$ such that
$U^{n}(i)\sim\psi^{(n)}_{n}$, where
\begin{align}
\psi^{(n)}_{n}(y^{n})=\begin{cases}
\frac{1}{|T_n(\psi_n)|},   &\text{if } y^{n}\in T_n(\psi_n)  \\
0,  &\text{otherwise.}
\end{cases}
\nonumber
\end{align}
(As usual, for simplicity we assume that $2^{nR}$ is an integer.) Let
$\hat{X}^{n}$ denote the output of the nearest neighborhood encoder:
$\hat{X}^{n}= \argmin\limits_{1\le i\le 2^{nR}} \rho_n(X^n,U^{n}(i))$. In case
of ties, we choose $U^{n}(i)$ with the smallest index $i$. The next lemma states
the intuitively clear fact that $\hat{X}^{n}$ is uniformly distributed on
$T_{n}(\psi_n)$.

\begin{lemma} 
\label{alemma1}
 $\hat{X}^{n}\sim \psi^{(n)}_{n} $.
\end{lemma}

The idea for this random coding scheme comes from \cite{ZaRo01} where an
infinite i.i.d.\ codebook $\{U^n(i)\}_{i=1}^{\infty}$ was considered and the
coding rate was defined as $(1/n)\log N_n$, where $N_n$ is the smallest index $i$
such that $\rho_n(X^n,U^n(i))\le D$. If the $U^n(i)$ are uniformly chosen from
the type class $T_n(\psi_n)$, then by Theorem~1 and Appendix~A and B of
\cite{ZaRo01}, $(1/n) \log N_n - I_m(\mu\|\psi_n,D)\to 0 $ in probability. 

Our scheme converts this variable-length random coding scheme into a fixed-rate
scheme by considering, for each blocklength $n$,  the finite codebook
\{$U^{n}(i)\}_{i=1}^{2^{nR}}$.  Letting $\rho_{\rm max} = \max_{x,y}
\rho(x,y)$, the expected distortion of our scheme is bounded as
\[
 E[\rho_n(X^{n},\hat{X}^{n})] \le D + \rho_{\rm max} \Pr\Bigl\{\,  \frac{1}{n}
 \log N_n > R \Bigr\} .
 \]
 Since $I_{m}(\mu \| \psi_{n}, D)\rightarrow I_{m}(\mu \| \psi, D)$ by the
 continuity of $I_m(\mu\|\psi,D)$ in $\psi$ (see \cite[Appendix A]{ZaRo01}), we
 have $R\ge
 I_m(\mu\|\psi_n,D)+\delta $ for some $\delta>0$ if $n$ is large enough. Thus the
 above bound  implies 
\begin{align}
\limsup_{n\to \infty} E[\rho_n(X^{n},\hat{X}^{n})]\leq D.
\label{aeq4}
\end{align}

Hence our random coding scheme has the desired rate and distortion as $n\to
\infty$. However, its output $\hat{X}^n$ has distribution $\psi^{(n)}_{n}$ instead of
the required $\psi^n$. The next lemma shows that the normalized
Kullback-Leibler divergence (relative entropy, \cite{CoTh06}) between
$\psi^{(n)}_{n} $ and $\psi^{n}$ asymptotically vanishes.

\begin{lemma}
$\dfrac{1}{n} \mathcal{D}(\psi^{(n)}_{n}\|\psi^{n})\rightarrow 0$ as $n\rightarrow\infty$.
\label{alemma2}
\end{lemma}

Let $\pi,\lambda \in \P(\sX)$. The optimal transportation cost $\hat{T}_n(\pi,\lambda)$
between $\pi $ and $\lambda$ (see, e.g., \cite{Vil09}) with cost function
$\rho_n$ is defined by
\begin{equation}
  \label{eq_optcoupl}
  \hat{T}_n(\pi,\lambda)=  \inf\bigl\{ E[\rho_n(U^n,V^n)]: \,
    U^n\sim \pi, \; V^n \sim \lambda \bigr\},
\end{equation}
where the infimum is taken over all joint distribution of pairs of random
vectors $(U^n,V^n)$ satisfying the given marginal distribution constraints. The
joint distribution achieving $\hat{T}_n(\pi,\lambda) $ as well as the resulting
pair $(U^n,V^n)$ are both called an optimal coupling of $\pi$ and $\lambda$.
Optimal couplings exist when $\sX$ is finite or $\sX=\R^n$, 
$\rho(x,y)=(x-y)^2$, and both $\pi$ and $\lambda$ both have finite second
moments \cite{Vil09}.

Now consider an optimal coupling $(\hat{X}^n,Y^n)$ of $\psi_n^{(n)}$ and $\psi^n$.  If
$Z_1$ and $Z_2$ are uniform random variables on $[0,1]$ such that 
$Z=(Z_1,Z_2)$ is independent of $X^n$, then the random code and optimal coupling
can be ``realized'' as $(U^n(1),\ldots,U^n(2^{nR}))= f_n(Z_1)$, $\hat{X}^n =
\hat{f}_n(X^n,Z_1)$,  and
$Y^n=g_n(\hat{X}^n,Z_2)$, where $f_n$, $\hat{f}_n$, and $g_n$ are suitable (measurable)
functions.  Combining random coding with optimal coupling this way gives rise
to a randomized
quantizer of type Model~2 whose output has the desired distribution
$\psi^n$ (see Fig.~\ref{afig1}).

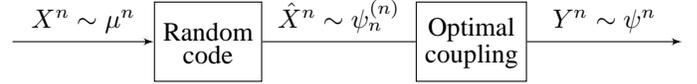
\begin{figure}[h]
\centering
\tikzstyle{int}=[draw, fill=white!20, minimum size=3em]
\tikzstyle{init} = [pin edge={to-,thin,black}]
\begin{tikzpicture}[node distance=3.5cm,auto,>=latex']
    \node [int] (a) {$\overset{\text{\normalsize Random}}{\text{code}}$};
    \node (b) [left of=a,node distance=2.6cm, coordinate] {a};
    \node [int] (c) [right of=a] {$\overset{\text{\normalsize Optimal}}{\text{\normalsize coupling}}$};
    \node [coordinate] (end) [right of=c, node distance=2.8cm]{};
    \path[->] (b) edge node {$X^n\sim\mu^n$} (a);
    \path[-] (a) edge node {$\hat{X}^n\sim\psi^{(n)}_{n}$} (c);
    \draw[->] (c) edge node {$Y^n\sim\psi^{n}$} (end) ;
\end{tikzpicture}
\caption{$D(R,\mu,\psi)$ achieving randomized quantizer scheme.}
\label{afig1}
\end{figure}

The next lemma uses Marton's inequality \cite{Mar96} to show that the
extra distortion introduced by the coupling step asymptotically vanishes.

\begin{lemma}
 \label{lem_coupling}
We have 
\[
\lim_{n\to \infty}   \hat{T}_n(\psi_n^{(n)},\psi^n) = 0
\]
and consequently
\[
\limsup_{n\to \infty} E\bigl[\rho_n(X^n,Y^n)\bigr] \le D.
\]
\end{lemma}

In summary, we have shown that there exists a sequence of Model~2 randomized
quantizers having rate $R=I_m(\mu\|\psi,D)+\varepsilon$ and asymptotic distortion
upper bounded by $D$ which satisfy the output distribution constraint $Y^n\sim
\psi^n$.  Since $\varepsilon>0$ is arbitrary, this completes the proof of the
achievability of $I_m(\mu\|\psi,D)$ (and the achievability of $D(\mu,\psi,R)$) for
finite source and reproduction alphabets.

\begin{remark}
We note that an obvious approach to achievability would be to generate a
codebook where the codewords have i.i.d.\ components drawn according to
$\psi$. However, the output distribution of the resulting the scheme would be
\emph{too far} from the desired $\psi^n$. In particular, such a scheme produces
output $\hat{X}^n$ whose empirical distribution (type) converges to a ``favorite
type'' which is typically different from $\psi$ \cite[Theorem~4]{ZaRo01}. As
well, the rate achievable with this scheme at distortion level $D$ is
\cite[Theorem~2]{YaKi98} 
\[
R= \min_{\psi' \in \P(\sY)} \bigl( I_m(\mu\|\psi', D) + \mathcal{D}(\psi'\|\psi)
\bigr) 
\]
which is typically strictly less than $I_m(\mu\|\psi,D)$. 
\end{remark}

Now let  $\sX=\sY=\R$, $\rho(x,y)=(x-y)^2$,  and assume that $\mu$ and
$\psi$ have finite second moments. We make use of the final alphabet case to
prove achievability for this continuous case. The following lemma provides the
necessary link between the two cases.

\begin{lemma}
\label{alemma4}
There exist a sequence $\{A_k\}$ of finite subsets of $\R$ and sequences
of probability measures $\{\mu_k\}$ and $\{\psi_k\}$, both supported on $A_k$,
such that 
\begin{itemize}
\item [(i)] $\hat{T}_1(\mu,\mu_k) \rightarrow 0$, $\hat{T}_1(\psi,\psi_k) \rightarrow 0$
  as $k\rightarrow\infty$;
\item [(ii)] For any $\varepsilon >0$ and 
  $D\geq0$ such that  $I_m(\mu\|\psi,D)<\infty$, we have 
 $I_m(\mu_k\|\psi_k,D+\varepsilon) \leq
  I_m(\mu\|\psi,D)$ for all $k$ large enough. 
\end{itemize}
\end{lemma}

Let $\mu_k^n$ and $\psi_k^n$ denote the $n$-fold products of $\mu_k$ and
$\psi_k$, respectively.  Definition \eqref{eq_optcoupl} of optimal coupling implies that
$\hat{T}_n(\mu^n,\mu_k^n)\leq \hat{T}_1(\mu,\mu_{k})$ and $\hat{T}_n(\psi^{n},\psi_{k}^{n})\leq
\hat{T}_1(\psi,\psi_{k})$.  Hence for any given $\varepsilon>0$ by Lemma~\ref{alemma4}
we can choose $k$ large enough such that  for all $n$,
\begin{align}
\hat{T}_n(\mu^n,\mu_{k}^{n}) \le \varepsilon \text{ and } 
\hat{T}_n(\psi^n,\psi_{k}^{n}) \leq \varepsilon,
\label{aeq7}
\end{align}
and also $I_m(\mu_k\|\psi_k,D+\varepsilon) \leq
I_m(\mu\|\psi,D)$.

Now for each $n$ define the following randomized quantizer:
\begin{itemize}
\item [(a)] Realize the optimal coupling between $\mu^{n}$ and $\mu_{k}^{n}$.
\item [(b)] Apply the randomized quantizer scheme for the finite alphabet case
  with common source and output alphabet $A_k$, source distribution
  $\mu_{k}^{n}$,  and output distribution
  $\psi_{k}^{n}$. Set the rate of the quantizer to
  $R=I_m(\mu\|\psi,D)+\varepsilon$. 
\item [(c)] Realize the optimal coupling between $\psi_{k}^{n}$ and $\psi^{n}$.
\end{itemize}

In particular, the optimal couplings are realized as follows: in (a) the source
$X^n\sim \mu^n$ is mapped to $X^n(k)\sim \mu_k^n$, which serves as the source in
(b), via $X^n(k)= \hat{f}_{n,k}(X^n,Z_3)$, and in (c) the output $Y^n(k)\sim
\psi_k^n$ of the scheme in (b) is mapped to $Y^n\sim \psi^n$ via $Y^n=
\hat{g}_{n,k}(Y^n(k),Z_4)$, where $Z_3$ and $Z_4$ are uniform randomization
variables that are independent of $X^n$. Thus the composition of these three
steps is a valid Model~2 randomized quantizer.

Since $R= I_m(\mu\|\psi,D)+\varepsilon$, in step (b) the asymptotic (in $n$)
distortion $D+\varepsilon$ can be achieved by Lemma~\ref{alemma4}(ii).  Using
(\ref{aeq7}) and the triangle inequality for the norm $\|V^n\|_2 \coloneqq
\bigl( \sum_{i=1}^nE[ V_i^2] \bigl)^{1/2} $ on $\R^n$-valued random vectors
having finite second moments, it is straightforward to show that the asymptotic
distortion of the overall scheme is upper bounded by $D + l(\varepsilon)$, where
$l(\varepsilon)\to 0$ as $\varepsilon\to 0$. Since $\varepsilon>0$ can be taken
to be arbitrarily small by choosing $k$ large enough, this completes the
achievability proof for the case $\sX=\sY=\R$ \qed

\section{Conclusion}

We investigated a general abstract model for randomized quantization that
provides a more suitable framework for certain optimal quantization problems
than the ones usually considered in the source coding  literature. In
particular, our model formalizes the notion of randomly picking a quantizer
from the set all $\emph{all}$ quantizers with a given number of output levels. 
Using this model, we proved the existence of an optimal
randomized vector quantizer under the constraint that the quantizer
output has a given distribution.  

Our results are mostly non-constructive and it is an open problem how to find
(or well approximate) such optimal quantizers.  A special case where a scalar
source has a density and the output distribution is constrained to be equal to
the source distribution was considered in \cite{LiKlKl10} and construction based
on dithered uniform quantization followed by a nonlinear mapping was
given. Although this construction is optimal in the limit of high resolution
($M\to \infty$), it is very likely suboptimal for any finite $M$. In general, it
would be interesting to better characterize optimal randomized quantizers in
Theorem~\ref{thmopt}, for example, by finding useful necessary conditions for
optimality. It would also be interesting to characterize the high-resolution
behavior of the distortion, which should be markedly different from the
classical case if the input and output distributions are not equal. Connections
between the output distribution-constrained lossy source codes studied in
Section~\ref{sec5} and the empirical distribution of good rate-distortion codes
(see, e.g., \cite{WeOr05} and references therein) are also worth
studying. Finally, a rigorous theory of randomized quantization paves the way for
interesting applications in signaling games in game theory \cite{SignalingGames}
and in stochastic networked control (see \cite{LiYu14} and \cite{BoMiSaTa05} for
applications of randomized quantization in real-time coding, and  \cite{YuLi12} and
\cite{YukselBasarBook} for quantizers and stochastic kernels viewed as
information structures in networked control).

\section*{APPENDIX}

\appsec

\subsection{\textbf{Proof of Lemma \ref{theo1}}}
\label{app1}
\noindent For  a fixed probability measure   $\mu$ on $\sX$ define 
 \begin{align} \Delta_{\mu}=\bigl\{v \in
  \P(\sX\times \sY):v(\,\cdot\,\times \sY)=\mu\} \nonumber
\end{align}
($\Delta_{\mu}$ is the set of all probability measures in
$\P(\sX\times\sY)$ whose $\sX$-marginal is $\mu$).  The following proposition,
due to Borkar \cite[Lemma 2.2]{Bor93}, gives a characterization of the extreme
points of $\Delta_{\mu}$.

\begin{proposition}
  $\Delta_{\mu}$ is closed and convex, and its set of extreme points
  $\Delta_{\mu,e}$ is a Borel set in $\P(\sX\times\sY)$. Furthermore,
  $v\in \Delta_{\mu,e}$ if and only if $v(dx\,dy)$ can be
  disintegrated as 
\begin{align}
v(dx\,dy)=Q(dy|x)\mu(dx) \nonumber
\end{align}
where $Q(\,\cdot\,|x)$ is a Dirac measure for $\mu$-a.e.\ $x$, i.e., there exists a
measurable function $f:\sX\rightarrow\sY$ such that
$Q(\,\cdot\,|x)=\delta_{f(x)}(\,\cdot\,)$ for $\mu$-a.e.\ $x$.
\label{prop1}
\end{proposition}

In fact, Borkar did not explicitly state Borel measurability of $\Delta_{\mu,e}$
in \cite{Bor93}, but the proof of \cite[Lemma 2.3]{Bor93}  clearly implies
this.

By Proposition \ref{prop1} it is clear that $v\in \Gamma_{\mu}(M)$ if and only if $v \in \Delta_{\mu,e}$ and its marginal on $\sY$ is supported on a set having at most $M$ elements, i.e., for some $L\leq M$ and $\{y_{1},\ldots,y_{L}\}\subset \sY$,
\begin{align}
v(\sX\times \{y_{1},\ldots,y_{L}\})=1. \nonumber
\end{align}
Let $\{y_{n}\}_{n\geq 1}$ be a countable dense subset of $\sY$ and define following subsets of $\Delta_{\mu,e}$:
\begin{align*}
  \Omega_{k}&=\bigcup_{n_{1}\geq 1,\ldots,n_{M}\geq 1}\biggl\{  v \in
  \Delta_{\mu,e}:v\Big(\sX\times
  \bigcup_{i=1}^{M}B(y_{n_{i}},1/k)\Big)=1\biggr\} \\
\intertext{and}
  \Sigma&=\bigcap_{k=1}^{\infty}\Omega_{k} \nonumber
\end{align*}
where $B(y,r)$ denotes the open ball in $\sY$ centered at $y$ having radius
$r$. Sets of the form
\begin{align}
\biggl\{v\in \P(\sX\times \sY):v\Big(\sX\times \bigcup_{i=1}^{M}B(y_{n_{i}},1/k)\Big)=1\biggr\}  \nonumber
\end{align}
are Borel sets by \cite[Proposition 7.25]{BeSh78}. Since $\Delta_{\mu,e}$
is a Borel set, $\Omega_k$ is a  Borel set for all $k$. Thus  $\Sigma$ is a Borel set in
$\P(\sX\times \sY)$. We will prove that  $\Sigma=\Gamma_{\mu}(M)$.

Since $\{y_{n}\}_{n\geq 1}$ is dense in $\sY$, for any $v\in\Gamma_{\mu}(M)$ and
$k\geq1$ there exist $\tilde{n}_{1},\ldots,\tilde{n}_{M}$ such that
$\supp(v(\sX\times\, \cdot\, ))\subset \bigcup
_{i=1}^{M}B(y_{\tilde{n}_{i}},1/k)$. Thus  $\Gamma_{\mu}(M) \subset \Omega_{k}$
for all $k$, implying $ \Gamma_{\mu}(M) \subset \Sigma$.

To prove the inclusion $\Sigma \subset \Gamma_{\mu}(M)$, let 
$v\in \Sigma$ and notice that for all $k$ there exist $n_{1}^{k},
n_{2}^{k}, \ldots, n_{M}^{k}$ such that
\begin{align}
v\Bigl(\sX\times \bigcup_{i=1}^{M}B(y_{n_{i}^{k}},1/k)\Bigr)=1. \nonumber
\end{align}
Let us define $K_{n}=\sX\times
\bigcap_{k=1}^{n}\bigcup_{i=1}^{M}B(y_{n_{i}^{k}},1/k)$. Clearly, $K_{n+1}
\subset K_{n}
$ and $v(K_{n})=1$, for all $n$. Letting 
\begin{align}
G=\bigcap_{k=1}^{\infty}\bigcup_{i=1}^{M} B(y_{n_{i}^{k}},1/k), \nonumber
\end{align}
we have $v(\sX\times G)=1$. If we can prove that $G$ has at most $M$ distinct
elements, then $v\in\Gamma_{\mu}(M)$. Assuming the contrary, there must exist
distinct $\{\hat{y}_{1}, \hat{y}_{2},\ldots,\hat{y}_{M},\hat{y}_{M+1}\} \subset
G$. Let $\varepsilon=\min \{\|\hat{y}_{i}-\hat{y}_{j}\|: i, j=1,\ldots, M+1,
i\neq j\}$. Clearly, for $\frac{1}{k}< \frac{\varepsilon}{4}$,
$\bigcup_{i=1}^{M}B(y_{n_{i}^{k}},1/k)$ cannot contain $\{\hat{y}_{1},
\hat{y}_{2},\ldots,\hat{y}_{M},\hat{y}_{M+1}\}$, a contradiction. Thus $G$ has
at most $M$ elements and we obtain  $\Sigma = \Gamma_{\mu}(M)$. \qed

\subsection{\textbf{Proof of Theorem \ref{newtheo1}}}
\label{app2}

We will need the following result which gives a necessary and sufficient condition for the
measurability of a mapping from a measurable space to $\P(\sE)$, where  $\sE$ is a Polish
space. It is proved for compact $\sE$ in \cite[Theorem 2.1]{DuFr64} and for
noncompact $\sE$ it is the corollary of \cite[Proposition 7.25]{BeSh78}.

\begin{theorem}
  Let $(\Omega,\mathcal{F})$ be a measurable space and let $\sE$ be a Polish space. A
  mapping $h:\Omega \rightarrow \mathcal{P}(\sE)$ is measurable if and only if
  the real valued functions $\omega \mapsto h(\omega)(A)$  from $\Omega$ to
  $[0,1]$  are measurable
  for all $A \in \B(\sE)$.
\label{theo2}
\end{theorem}

For any $(\sq,\nu)$ define $f: \R^m \to \Gamma_{\mu}(M)$ by $f(z) =
\delta_{\sq(x,z)}(dy)\mu(dx)$.  By Theorem \ref{theo2}, $f$ is measurable if and
only if the mappings $z\mapsto\int\delta_{\sq(x,z)}(C_{x})\mu(dx)$ are
measurable for all $C\in\B(\sX\times\sY)$, where $C_{x}=\{y:(x,y)\in
C\}$. Observe that $\delta_{\sq(x,z)}(C_{x})$ is a measurable function of
$(x,z)$ because $\{(x,z)\in \sX\times\sZ:
\delta_{\sq(x,z)}(C_x)=1\}=\{(x,z)\in\sX\times\sZ: (x,\sq(x,z))\in C\}$. By
\cite[Theorem 18.3]{Bil95} $\int\delta_{\sq(x,z)}(C_{x})\mu(dx)$ is measurable
as well. Hence $f$ is measurable.

Thus we can define the probability measure $P$ supported on
$\Gamma_{\mu}(M)$ by $P=\nu \circ f^{-1}$ (i.e., $P(B)=\nu(f^{-1}(B))$ for any
Borel set $B\subset \Gamma_{\mu}(M)$).  Then, for the corresponding $v_P$ we have
$(X,Y)\sim v_P$, i.e., for $C\in\B(\sX\times\sY)$,
\begin{align}
\Pr\bigl\{\bigl(X,\sq(X,Z)\bigr)\in C\bigr\}&=\int_{\sZ}\int_{\sX} \delta_{\sq(x,z)}(C_{x}) \mu(dx) \nu(dz)\nonumber \\
&=\int_{\sZ} f(z)(C)\nu(dz) \nonumber \\
&=\int_{\Gamma_{\mu}(M)}v(C)P(dv)\nonumber \\
&=v_{P}(C). \nonumber
\end{align}

Conversely, let $v_P$ be defined as in (\ref{neweq2}) with $P$ supported on
$\Gamma_{\mu}(M)$, i.e., $v_P=\int_{\Gamma_{\mu}(M)} v P(dv)$.  Define the
mapping $ \Gamma_{\mu}(M) \ni v\mapsto q_{v}$, where $q_v$ is the $\mu$-a.e.\
defined quantizer in $\Q_M$, giving
$v(dx\,dy)=\mu(dx)\delta_{q_{v}(x)}(dy)$. Since $\Gamma_{\mu}(M)$ is an
uncountable Borel space, there is a measurable bijection (Borel isomorphism) $g
: \R^m \to \Gamma_{\mu}(M)$ between $\R^m$ and $\Gamma_{\mu}(M)$ \cite{Dud89}.
Now define $\sq$ by $\sq(x,z)= q_{g(z)}(x)$ and let $\nu=P\circ g$.  Then for
all $z$, $\sq(\,\cdot\,,z)$ is a $\mu$-a.e. defined $M$-level
quantizer. However, it is not clear whether $\sq(x,z)$ is measurable. Therefore
we will construct another measurable function $\tilde{\sq}(x,z)$ such that
$\tilde{\sq}(\,\cdot\,,z)$ is an $M$-level quantizer and
$\tilde{\sq}(\,\cdot\,,z)=\sq(\,\cdot\,,z)$ $\mu$-a.e., for all $z$. Then we
will prove that $(X,Y)=(X,\tilde{\sq}(X,Z))\sim v_{p}$ where $Z\sim\nu$.  Define
the stochastic kernel on $\sX\times \sY$ given $\Gamma_{\mu}(M)$ as
\begin{align}
\gamma(dx\,dy|v)=v(dx\,dy). \nonumber
\end{align}
Clearly, $\gamma$ is well defined because $\Gamma_{\mu}(M)$ is a Borel subset of $\P(\sX\times \sY)$. Observe that for each $v\in\Gamma_{\mu}(M)$, we have
\begin{align}
\gamma(C|v)=\int_{\sX}\delta_{q_{v}(x)}(C_{x})\mu(dx)
\label{oldeq1}
\end{align}
for $C\in\B(\sX\times\sY)$. Furthermore, by \cite[Proposition 7.27]{BeSh78}
there exists a stochastic kernel $\eta(dy|x,v)$ on $\sY$ given
$\sX\times\Gamma_{\mu}(M)$ which satisfies for all $C\in\B(\sX\times\sY)$ and
$v\in\Gamma_{\mu}(M)$, 
\begin{align}
\gamma(C|v)=\int_{\sX} \eta(C_{x}|x,v) \mu(dx).
\label{oldeq2}
\end{align}
Since $\B(\sY)$ is countably generated by the separability of $\sY$, for any
$v\in\Gamma_{\mu}(M)$ we have $\eta(\,\cdot\,|x,v)=\delta_{q_{v}(x)}(\,\cdot\,)$
$\mu$-a.e. by (\ref{oldeq1}) and (\ref{oldeq2}). Since $\eta$ is a stochastic
kernel, it can be represented as a measurable function from
$\sX\times\Gamma_{\mu}(M)$ to $\P(\sY)$, i.e.,
\begin{align}
\eta:\sX\times\Gamma_{\mu}(M)\rightarrow\P(\sY). \nonumber
\end{align}
Define $\P_{1}(\sY)=\{\psi\in\P(\sY): \psi(\{y\})=1 \text{ for some }
y\in\sY\}$. $\P_{1}(\sY)$ is a closed (thus measurable) subset of $\P(\sY)$ by
\cite[Lemma 6.2]{Par67}. Hence, $\sM:=\eta^{-1}(\P_{1}(\sY))$ is also
measurable. Observe that for any $v\in\Gamma_{\mu}(M)$ we have
$\sM_{v}:=\{x\in\sX: (x,v)\in \sM\}\supset \{x\in\sX:
\eta(\,\cdot\,|x,v)=\delta_{q_{v}(x)}(\,\cdot\,)\}$. Thus $\mu(\sM_{v})=1$ for all
$v\in\Gamma_{\mu}(M)$, which implies $\mu\otimes P \bigl(\sM\bigr)=1$. Define the
function $\tilde{q}_{v}$ from $\sX\times\Gamma_{\mu}(M)$ to $\sY$ as 
\begin{align}
\tilde{q}_{v}(x)=\begin{cases}
\tilde{y},  &\text{ if } (x,v)\in \sM, \text{ where } \eta(\{\tilde{y}\}|x,v)=1,  \\
y,  &\text{ otherwise, }
\end{cases}
\nonumber
\end{align}
where $y$ is fixed. By construction, $\tilde{q}_{v}(x)=q_{v}(x)$ $\mu$-a.e., for all $v\in\Gamma_{\mu}(M)$. For any $C\in\B(\sY)$ we have
\begin{eqnarray*}
 \lefteqn{\tilde{q}_{v}^{-1}(C)} \\
&=&\{(x,v)\in\sX\times\Gamma_{\mu}(M): \tilde{q}_{v}(x) \in C\} \nonumber \\
&=& \{(x,v)\in \sM: \tilde{q}_{v}(x)\in C\}\cup \{(x,v)\in \sM^{c}: \tilde{q}_{v}(x)\in C\}. \nonumber
\end{eqnarray*}
Clearly $\{(x,v)\in \sM^{c}: \tilde{q}_{v}(x)\in C\}=\sM^{c} \text{ or
}\emptyset$ depending on whether or not $y$ is an element of $C$. Hence,
$\tilde{q}_{v}^{-1}(C)\in\B(\sX\times\Gamma_{\mu}(M))$ if $\{(x,v)\in \sM:
\tilde{q}_{v}(x)\in C\} \in\B(\sX\times\Gamma_{\mu}(M))$. But $\{(x,v)\in \sM:
\tilde{q}_{v}(x)\in C\}=\{(x,v)\in \sM: \eta(C|x,v)=1\}$ which is in
$\B(\sX\times\Gamma_{\mu}(M))$ by the measurability of $
\eta(C|\,\cdot\,,\,\cdot\,)$. Thus, $\tilde{q}$ is a measurable function from
$\sX\times\Gamma_{\mu}(M)$ to $\sY$.

Let us define $\tilde{\sq}$ as $\tilde{\sq}(x,z)=\tilde{q}_{g(z)}(x)$. By the measurability of $g$ it is clear that $\tilde{\sq}$ is measurable. In addition, for any $z\in\sZ$
$\tilde{\sq}(\,\cdot\,,z)$ is an $M$-level quantizer which is $\mu$-a.e. equal to $\sq(\,\cdot\,,z)$. Finally, if $Z\sim \nu$ is independent of $X$ and $Y=\tilde{\sq}(X,Z)$, then $(X,Y)\sim v_P$, i.e.,

\begin{eqnarray*}
\lefteqn{ \Pr\Bigl\{\bigl(X,\tilde{\sq}(X,Z)\bigr)\in C\Bigr\}}\\
&=&\int_{\sZ}\int_{\sX}
\delta_{\tilde{\sq}(x,z)}(C_{x}) \mu(dx)\nu(dz) \nonumber \\* 
&=&\int_{\Gamma_{\mu}(M)}\int_{\sX} \delta_{\tilde{q}_{v}(x)}(C_{x}) \mu(dx)P(dv) \nonumber \\
&=&\int_{\Gamma_{\mu}(M)}\int_{\sM_{v}} \eta(C_{x}|x,v) \mu(dx)P(dv)\nonumber \\
&=&\int_{\Gamma_{\mu}(M)}\gamma(C|v) P(dv)\nonumber\\
&=&\int_{\Gamma_{\mu}(M)}v(C) P(dv)\nonumber\\
&=&v_{p}(C). \qquad \qquad \qquad \qquad \qquad \qquad \qquad  \text{\qed}
\end{eqnarray*}

\subsection{\textbf{Proof of Theorem~\ref{newtheo2}}}
\label{app2a}

If  $(e,d,\nu)$ is a  Model~1 randomized quantizer,  then setting
$\sq(x,z)=d(e(x,z),z)$  defines a 
Model~2 randomized quantizer $(\sq,\nu)$ such that the joint distributions of
their inputs and outputs coincide.

Conversely, let $(\sq,\nu)$ be a Model~2 randomized quantizer. It is obvious
that $\sq$ can be decomposed into an encoder $e:\sX\times \sZ\to \{1,\ldots,M\}$
and decoder $d: \{1,\ldots,M\}\times \sZ \to \sY$ such that $d(e(x,z),
z)=\sq(x,z)$ for all $x$ and $z$. The difficulty lies in showing that this can
be done so that the resulting $e$ and $d$ are measurable. In fact, we instead
construct measurable $e$ and $d$ whose composition  is $\mu\otimes \nu$-a.e.\   equal to
$\sq$, which  is sufficient to imply the theorem.

Let $(\sq,\nu)$ be a Model 2 randomized quantizer. Since $\R^n$ and $[0,1]$ are
both uncountable Borel spaces, there exists  a Borel isomorphism
$f:\R^n\rightarrow[0,1]$ \cite{Dud89}. Define
$\hat{\sq}: \sX\times\sY\to [0,1]$ by $\hat{\sq}= f\circ\sq$. Hence, $\hat{\sq}$
is measurable and, for any
fixed $z$, $\hat{\sq}(\cdot,z)$ is an $M$-level quantizer from $\sX$ to
$[0,1]$. Also note that $\sq=f^{-1}\circ \hat{\sq}$.  

Now for any fixed $z\in \sZ$ consider only those output points of $\hat{\sq}(\cdot,z)$ that
occur with \emph{positive} $\mu$ probability and order these according to their
magnitude from the smallest to the largest.  For $i=1,\ldots,M$ let the function
$f_i(z)$ take the value of the $i$th smallest such output point.  If there is no
such value, let $f_i(z)= 1$.  We first prove that all the  $f_i$ are measurable
and then  define the encoder and the decoder in terms of these
functions.

Observe that for any
$a\in[0,1]$, by definition 
\[
\{z\in\sZ:f_1(z)\leq a\}= \Big\{z\in\sZ:
\int_{\sX} \delta_{\hat{\sq}(x,z)}([0,a])\mu(dx)>0\Big\},
\]
where the set on the right hand side is a Borel set by Fubini's theorem. Hence,
$f_1$ is a measurable function. Define $E_1=\{(x,z)\in\sX\times\sZ:
\hat{\sq}(x,z)-f_1(z)=0\}$, a Borel set. Letting  $E_{1,z}= \{x\in \sX :
(x,z)\in E_1\}$ denote the 
$z$-section of $E_1$, for any $a\in[0,1)$ we have 
\begin{eqnarray*}
\lefteqn{  \{z\in\sZ:f_2(z)\leq a\}}\quad \\
&=&\Big\{z\in\sZ: \int_{\sX\setminus
  E_{1,z}} \delta_{\hat{\sq}(x,z)}([0,a])\mu(dx)>0\Big\},
\end{eqnarray*}
and thus $f_2$ is measurable. Continuing in this fashion, we define 
the Borel sets $E_i=\{(x,z): \hat{\sq}(x,z)-f_i(z)=0\}$ and write,  
for any $a\in[0,1)$,
\begin{eqnarray*}
\lefteqn{ \{z\in\sZ:f_i(z)\leq a\} }\quad  \\
&=&\Big\{z\in\sZ: \int_{\sX\setminus
  \bigcup_{j=1}^{i-1} E_{i,z}} \delta_{\hat{\sq}(x,z)}([0,a])\mu(dx)>0\Big\},
\end{eqnarray*} 
proving that $f_i$ is
measurable  for all  $i=1,\ldots,M$.

Define
\begin{align*}
N&=\big\{(x,z)\in\sX\times\sZ: \hat{\sq}(x,z)\neq f_i(z)\text{ for all
  $i=1,\ldots,M$} \big\} \\
 &= \sX\times\sZ \setminus  \bigcup_{i=1}^M E_i .
\end{align*}
Clearly, $N$ is a Borel set and $\mu\otimes\nu(N)=0$ by Fubini's theorem and the
definition of $f_1,\ldots,f_M$. Now we can define
\begin{align}
e(x,z)&=\sum_{i=1}^{M} i\, 1_{\{\hat{\sq}(x,z)=f_i(z)\}} + M \, 1_{N}(x,z)
\nonumber \\
\intertext{and}
d(i,z)&=\sum_{j=1}^{M} f^{-1}\circ f_j(z) 1_{\{i=j\}}, \nonumber 
\end{align}
where $1_{B}$ denotes the indicator of event (or set) $B$. The measurability of
$\hat{\sq}$ and $f$, $f_1,\ldots,f_M$ implies that $e$ and $d$ are
measurable. Since $d(e(x,z),z)=\hat{\sq}(x,z)$ $\mu\otimes\nu$-a.e. by
construction, this completes the proof. \qed 

\subsection{\textbf{Proof of Theorem \ref{thmopt}}}
\label{app3}

\noindent\textit{I) \textbf{Proof under Assumption~1}}

To simplify the notation we redefine the reconstruction alphabet as
$\sY=B$, so that $\sY$ is a compact subset of $\R^n$.  
It follows from the
continuity of $\rho$ that $L$ is lower semicontinuous on
$\mathcal{P}(\sX\times\sY)$ for the weak topology (see, e.g.,
\cite[Lemma~4.3]{Vil09}). Hence, to show the existence of a minimizer for problem
\textbf{(P1)} it would suffice to prove that
$\Gamma_{\mu,\psi}^{\rm R}(M)=\Gamma_{\mu}^{\rm R}(M)\cap \Gamma_{\mu,\psi}$ is
compact. It is known that $\Gamma_{\mu,\psi}$ is compact
\cite[Chapter 4]{Vil09}, but unfortunately $\Gamma_{\mu}(M)$ is not
closed \cite{YuLi12} and it seems doubtful that $\Gamma_{\mu}^{\rm R}(M)$ is
compact. Hence, we will develop a different argument which is based
on optimal transport theory. We will first give the proof under
Assumption~1; the proof under Assumption~2 then follows via a one-point
compactification argument.

Let $\mathcal{P}_{M}(\sY) =\{\psi_0 \in \mathcal{P}(\sY): |\supp(\psi_0)| \leq
M\}$ be the set of discrete distributions with $M$ atoms or less on $\sY$.

\begin{lemma}
$\P_{M}(\sY)$ is compact in $\P(\sY)$.
\label{lemma2}
\end{lemma}
\begin{proof}
  Let $\{\psi_{n}\}$ be an arbitrary sequence in $\P_{M}(\sY)$. Each $\psi_{n}$
  can be represented by points $(y_{1}^{n},\ldots,y_{M}^{n})=y^{n} \in \sY^{M}$
  and $(p_{1}^{n},\ldots,p_{M}^{n})=p^{n} \in K_{s}$, where
  $K_{s}=\{(p_1,\ldots,p_M)\in\R^M: \sum_{i=1}^{M} p_i=1, p_i\geq0\}$ is the
  probability simplex in $\R^{M}$. Let $w_n=(y^{n},p^{n})$. Since $\sY^{M}\times
  K_{s}$ is compact, there exists a subsequence $\{w^{n_{k}}\}$ converging to
  some $w$ in $\sY^{M}\times K_{s}$. Let $\psi$ be the probability measure in
  $\P_{M}(\sY)$ which is represented by $w$. It straightforward to show that
  $\psi$ is a weak limit of $\{\psi^{n_{k}}\}$. This completes the proof.
\end{proof}

Define
\[
\hat{\Gamma}_{\mu}(M)=\bigcup_{\psi_0\in \P_M(\sY)} \bigl\{ \hat{v}\in \Gamma_{\mu,\psi_0}:  L(\hat{v})
=\min_{v\in \Gamma_{\mu,\psi_0}} L(v) \bigr\}. 
\]
The elements of $\hat{\Gamma}_{\mu}(M)$ are the probability measures
which solve the optimal transport problem (see, e.g.,
\cite{Vil09}) for fixed input marginal $\mu$ and some output marginal
$\psi_0$ in $\mathcal{P}_{M}(\sY)$. At the end of this proof Lemma~\ref{lem_meas} shows that
$\hat{\Gamma}_{\mu}(M)$ is a Borel set. Let
$\hat{\Gamma}_{\mu}^{\rm R}(M)$ be the randomization of
$\hat{\Gamma}_{\mu}(M)$, obtained by replacing $\Gamma_{\mu}(M)$ with
$\hat{\Gamma}_{\mu}(M)$ in (\ref{neweq4}). Define the optimization
problem \textbf{(P2)} as
\begin{align}
\textbf{(P2)} \text{                         }&\text{minimize } L(v)
\nonumber \\
&\text{subject to  } v \in \hat{\Gamma}_{\mu,\psi}^{\rm R}(M), \nonumber
\end{align}
where $\hat{\Gamma}_{\mu,\psi}^{\rm R}(M)=\hat{\Gamma}_{\mu}^{\rm R}(M) \cap
  \Gamma_{\mu,\psi}$.

\begin{proposition}
  For any $v^{*}\in\Gamma_{\mu,\psi}^{\rm R}(M)$ there exists
  $\hat{v}\in\hat{\Gamma}_{\mu,\psi}^{\rm R}(M)$ such that $L(v^{*})\geq
  L(\hat{v})$. Hence, the distortion of any minimizer in \emph{\textbf{(P2)}} is less
  than or equal to the distortion of a minimizer in \emph{\textbf{(P1)}}.
\label{prop5}
\end{proposition}
To prove Proposition \ref{prop5} we need the following lemma.

\begin{lemma}
Let $P$ be a probability measure on $\Gamma_{\mu}(M)$. Then there exists a
measurable mapping $f:\Gamma_{\mu}(M)\to  \hat{\Gamma}_{\mu}(M)$ such that $v(\sX\times\,\cdot\,)=f(v)(\sX\times\,\cdot\,)$ and $L(v)\geq L(f(v))$, $P$-a.e.
\label{lemma4}
\end{lemma}
\begin{proof}
Define the projections  $f_{1}:\Gamma_{\mu}(M)\rightarrow \P_{M}(\sY)$ and
$f_{2}:\hat{\Gamma}_{\mu}(M)\rightarrow \P_{M}(\sY)$ by
$f_{1}(v)=v(\sX\times\,\cdot\,)$, $f_{2}(v)=v(\sX\times\,\cdot\,)$.
Note that $f_{1}$ is continuous and $f_{2}$ is continuous and onto. Define
$\tilde{P}=P\circ f_1^{-1}$ on $\P_{M}(\sY)$. By Yankov's lemma \cite[Appendix
3]{Dyk79} there exists a mapping $g$ from $\P_{M}(\sY)$ to
$\hat{\Gamma}_{\mu}(M)$ such that $f_{2}(g(\psi))=\psi$ $\tilde{P}$-a.e. Then,
it is straightforward to show that $f=g\circ f_{1}$ satisfies conditions
$v(\sX\times\,\cdot\,)=f(v)(\sX\times\,\cdot\,)$ and $L(v)\geq
L(f(v))$, $P$-a.e.
\end{proof}

\begin{proof}[Proof of Proposition \ref{prop5}]
Let $v^{*}\in\Gamma_{\mu,\psi}^{\rm R}(M)$, i.e.,
\begin{align}
v^{*}=\int_{\Gamma_{\mu}(M)} v P(dv) \text{ and } v^{*}(\sX\times\,\cdot\,)=\psi. \nonumber
\end{align}
By Lemma \ref{lemma4} there exists
$f:\Gamma_{\mu}(M)\rightarrow\hat{\Gamma}_{\mu}(M)$ such that
$v(\sX\times\,\cdot\,)=f(v)(\sX\times\,\cdot\,)$ and $L(v)\geq L(f(v))$,
$P$-a.e. Define $\tilde{P}=P\circ f^{-1}\in \P(\hat{\Gamma}_{\mu}(M))$ and
$\hat{v}=\int_{\hat{\Gamma}_{\mu}(M)} v
\tilde{P}(dv)\in\hat{\Gamma}_{\mu}^{\rm R}(M)$. We have
\begin{align*}
  L(v^{*})&=\int_{\Gamma_{\mu}(M)}L(v) P(dv)\geq \int_{\Gamma_{\mu}(M)}L(f(v))
  P(dv)\\
& =\int_{\hat{\Gamma}_{\mu}(M)} L(v) \tilde{P}(dv)=L(\hat{v}) 
\intertext{as well as}
  v^{*}(\sX\times\,\cdot\,)&=\int_{\Gamma_{\mu}(M)}v(\sX\times\,\cdot\,)
  P(dv)\\
&=\int_{\Gamma_{\mu}(M)}f(v)(\sX\times\,\cdot\,)
  P(dv)\\
& =\int_{\hat{\Gamma}_{\mu}(M)} v(\sX\times\,\cdot\,)
  \tilde{P}(dv)=\hat{v}(\sX\times\,\cdot\,).\nonumber
\end{align*}
This completes the proof.
\end{proof}

Recall the set $\Delta_{\mu}$ and its set of its extreme points $\Delta_{\mu,e}$
from Proposition \ref{prop1}. It is proved in \cite{Bor93} and \cite{Bor91} that
any $\tilde{v}\in\Delta_{\mu}$ can be written as
$\tilde{v}=\int_{\Delta_{\mu,e}}v P(dv)$ for some $P\in\P(\Delta_{\mu,e})$. By
Proposition \ref{prop1} we also have
$\Gamma_{\mu}(M)\subset\Delta_{\mu,e}$. The following lemma is based
on these two facts.

\begin{lemma}
Let $\tilde{v}\in\Delta_{\mu}$ which is represented as $\tilde{v}=\int_{\Delta_{\mu,e}}v P(dv)$. If $\tilde{v}(\sX\times\,\cdot\,)\in\P_{M}(\sY)$, then $P(\Gamma_{\mu}(M))=1$.
\label{lemma5}
\end{lemma}
\begin{proof}
Since $\tilde{v}(\sX\times\,\cdot\,)\in\mathcal{P}_{M}(\sY)$, there exist a
finite set $B\subset Y$ having $M^{'}\leq M$ elements  such that
$\tilde{v}(\sX\times B)=1$. We have 
\begin{align*}
  \tilde{v}(\sX\times B)=&{}\int_{\Delta_{\mu,e}}v(\sX\times B)P(dv)  \\
=  &{}\int_{\Delta_{\mu,e}\setminus
    \Gamma_{\mu}(M)}v(\sX\times B)P(dv) \\
& \mbox{}
  +\int_{\Gamma_{\mu}(M)}v(\sX\times B)P(dv). \nonumber
\end{align*}
Since $v(\sX\times B)<1$ for all $v\in\Delta_{\mu,e}\setminus \Gamma_{\mu}(M)$, we obtain $P(\Gamma_{\mu}(M))=1$.
\end{proof}

Lemma \ref{lemma5} implies $\hat{\Gamma}_{\mu}(M)\subset\Gamma_{\mu}^{\rm R}(M)$ because $v(\sX\times\,\cdot\,)\in\P_{M}(\sY)$ when $v\in\hat{\Gamma}_{\mu}(M)$. Define $h:\P(\Gamma_{\mu}(M))\rightarrow \Delta_{\mu}$ as follows:
\begin{align}
\label{eq_hmap}
h(P)(\,\cdot\,)=\int_{\Gamma_{\mu}(M)} v(\,\cdot\,) P(dv). 
\end{align}
It is clear that the range of $h$ is $\Gamma_{\mu}^{\rm R}(M)\subset \Delta_{\mu}$.

\begin{lemma}
$h$ is continuous.
\label{lemma6}
\end{lemma}
\begin{proof}
Assume $\{P_{n}\}$ converges weakly to $P$ in $\P(\Gamma_{\mu}(M))$. Then, for
any  continuous and bounded real function $f$ on $\sX\times \sY$
\begin{eqnarray*}
\lefteqn{\lim_{n\to \infty} \int_{\Gamma_{\mu}(M)}\int_{\sX\times\sY} f(x,y) v(dx\, dy)
P_{n}(dv)}\qquad \\
& = &\int_{\Gamma_{\mu}(M)}\int_{\sX\times\sY} f(x,y) v(dx\, dy) P(dv)
\end{eqnarray*}
if the mapping $v \mapsto \int_{\sX\times\sY} f(x,y) v(dx\, dy)$ is continuous and
bounded on $\Gamma_{\mu}(M)$. Clearly this mapping is  continuous by the definition of
weak convergence and bounded by the boundedness of $f$. Thus
\begin{align}
 \int_{\Gamma_{\mu}(M)}v P_{n}(dv) \to  \int_{\Gamma_{\mu}(M)} v P(dv) \nonumber
\end{align}
weakly, completing the proof.
\end{proof}

Since $\hat{\Gamma}_{\mu}(M)\subset\Gamma_{\mu}^{\rm R}(M)$, we have
$\P^{\rm opt}(\Gamma_{\mu}(M)):=h^{-1}(\hat{\Gamma}_{\mu}(M))\subset
\P(\Gamma_{\mu}(M))$, which is measurable by the measurability of
$\hat{\Gamma}_{\mu}(M)$ and $h$. Let $g:\P^{\rm opt}(\Gamma_{\mu}(M))\rightarrow
\hat{\Gamma}_{\mu}(M)$ be the restriction of $h$ to
$\P^{\rm opt}(\Gamma_{\mu}(M))$. Clearly $g$ is measurable and onto. By Yankov's
lemma \cite{Dyk79} for any probability measure $P$ on $\hat{\Gamma}_{\mu}(M)$
there exists a measurable mapping $\varphi:\hat{\Gamma}_{\mu}(M)\rightarrow
\P^{\rm opt}(\Gamma_{\mu}(M))$ such that $g(\varphi(\hat{v}))=\hat{v}$
$P$-a.e. In addition, since $\varphi(\hat{v})\in g^{-1}(\hat{v})$
$P$-a.e., we have
\begin{align}
L(\hat{v})&=\int_{\Gamma_{\mu}(M)}L(v)\varphi(\hat{v})(dv) \label{eq2}
\end{align}
and
\begin{align}
\hat{v}(\sX\times\,\cdot\,)&=\int_{\Gamma_{\mu}(M)}v(\sX\times\,\cdot\,)\varphi(\hat{v})(dv)
\label{eq3}
\end{align}
$P$-a.e. Define the stochastic kernel $\Pi(dv|\hat{v})$ on $\Gamma_{\mu}(M)$ given $\hat{\Gamma}_{\mu}(M)$ as
\begin{align}
\Pi(dv|\hat{v})=\varphi(\hat{v})(dv).
\label{eq1}
\end{align}
Since $\varphi$ is measurable, $\Pi(dv|\hat{v})$ is well defined. Observe that
both $\varphi$ and $\Pi(dv|\hat{v})$ depend on the probability measure $P\in \hat{\Gamma}_{\mu}(M)$.

\begin{proposition}
If \emph{\textbf{(P2)}} has a minimizer $v^{*}$, then we can find $\bar{v} \in \Gamma_{\mu,\psi}^{\rm R}(M)$
such that $L(\bar{v})=L(v^{*})$, implying that $\bar{v}$ is a minimizer for
\emph{\textbf{(P1)}}.\label{theo6}
\end{proposition}
\begin{proof}
  $v^{*}$ can be written as
  $v^{*}=\int_{\hat{\Gamma}_{\mu}(M)}\hat{v}P(d\hat{v})$. Consider the
  stochastic kernel $\Pi(dv|\hat{v})$ defined in (\ref{eq1}). Composing  $P$
  with  $\Pi$ we obtain a probability measure $\Lambda$ on
  $\hat{\Gamma}_{\mu}(M)\times \Gamma_{\mu}(M)$ given by
\begin{align}
\Lambda(d\hat{v}\, dv)=P(d\hat{v})\Pi(dv|\hat{v}).
\label{eq4}
\end{align}
Let $\tilde{P} = \Lambda(\hat{\Gamma}_{\mu}(M)\times \,\cdot\,)\in \P(
\Gamma_{\mu}(M))$. Define the randomized quantizer
$\bar{v}\in\Gamma_{\mu}^{\rm R}(M)$ as $\bar{v}=\int_{\Gamma_{\mu}(M)}
v\tilde{P}(dv)$. We show that $L(v^{*})=L(\bar{v})$ and $v^{*}(\sX \times
\,\cdot\,)=\bar{v}(\sX\times \,\cdot\,)$ which will complete the proof. We have
\begin{align}
L(v^{*})&=\int_{\hat{\Gamma}_{\mu}(M)}
L(\hat{v})P(d\hat{v}) \nonumber \\
&=\int_{\hat{\Gamma}_{\mu}(M)} \int_{\Gamma_{\mu}(M)}
L(v)\varphi(\hat{v})(dv) P(d\hat{v}) \text{                 \ \  (by
  (\ref{eq2}))} \nonumber \\ 
&=\int_{\hat{\Gamma}_{\mu}(M)\times\Gamma_{\mu}(M)} L(v)\Lambda(d\hat{v}\, dv)
\quad \text{(by \eqref{eq1})} \nonumber \\*
&=\int_{\Gamma_{\mu}(M)} L(v)\tilde{P}(dv)=L(\bar{v}). \nonumber
\end{align}
Similarly, 
\begin{align*}
  v^{*}(\sX\times\,\cdot\,)&=\int_{\hat{\Gamma}_{\mu}(M)}
  \hat{v}(\sX\times\,\cdot\,)P(d\hat{v})\\
& =\int_{\hat{\Gamma}_{\mu}(M)} \int_{\Gamma_{\mu}(M)}\!\!\!\!
v(\sX\times\,\cdot\,)\varphi(\hat{v})(dv) P(d\hat{v}) \text{   (by (\ref{eq3}))} \nonumber \\
  &=\int_{\hat{\Gamma}_{\mu}(M)\times\Gamma_{\mu}(M)}
  v(\sX\times\,\cdot\,)\Lambda(d\hat{v}\,  dv) \quad \text{(by \eqref{eq1})}    \\
  &=\int_{\Gamma_{\mu}(M)}
  v(\sX\times\,\cdot\,)\tilde{P}(dv)=\bar{v}(\sX\times\,\cdot\,). \nonumber
\end{align*}
By Proposition \ref{prop5}, $\bar{v}$ is a minimizer for $\textbf{(P1)}$.
\end{proof}

Hence, to prove the existence of a minimizer for \textbf{(P1)} it is enough
prove the existence of a minimizer for \textbf{(P2)}. Before proceeding to the
proof we need to define the optimal transport problem. Optimal transport problem
for marginals $\pi\in\P(\sX)$, $\lambda\in\P(\sY)$ and cost function
$c:\sX\times\sY\rightarrow[0,\infty]$ is defined as:
\begin{equation}
  \label{opttr}  
\begin{split}
&\text{minimize } \int_{\sX\times\sY} c(x,y) v(dx\,dy)  \\
&\text{subject to  } v \in \Gamma_{\pi,\lambda}. 
\end{split}
\end{equation}

The following result is about the structure of the optimal $v$ in
(\ref{opttr}). It uses the concept of $c$-cyclically monotone sets
\cite[Definition~5.1]{Vil09}. A set $B\subset \sX\times \sY$ is said to be
$c$-cyclically monotone if for any $N\ge 1$ and pairs
$(x_1,y_1),\ldots,(x_N,y_N)$ in $B$, the following inequality holds:
\[
\sum_{i=1}^N c(x_i,y_i) \le \sum_{i=1}^N c(x_i,y_{i+1}),
\]
where $y_{N+1}\coloneqq y_1$.

Informally, when $v\in \Gamma_{\pi,\lambda}$ is concentrated on a $c$-cyclically
monotone set, then its cost cannot be improved by local perturbations; see the
discussion in \cite[Chapter 5]{Vil09}. The following result shows that 
an optimal $v$ must concentrate on a $c$-cyclically monotone set.

\begin{proposition}[\protect{\cite[Theorem 1.2]{Pra08}, \cite[Theorem 5.10]{Vil09}}]
  Let $c:\sX\times\sY\rightarrow[0,\infty]$ be continuous. If $v \in
  \Gamma_{\pi,\lambda}$ is a solution to the optimal transport problem
  (\ref{opttr}) and $\int_{\sX\times\sY} c(x,y) v(dx\, dy)<\infty$, then $v$ is
  concentrated on some $c$-cyclically monotone set.
\label{prop7}
\end{proposition}

For any $K\subset\P(\sX)$ and $S\subset\P(\sY)$ define
$\Xi_{K,S}\subset\P(\sX\times\sY)$ as the set of probability measures which are
concentrated on some $c$-cyclically monotone set and solve (\ref{opttr}) for
some $\pi\in K$, $\lambda\in S$. The following result is a slight modification
of \cite[Corollary~5.21]{Vil09}. 

\begin{proposition}
If $K$ and $S$ are compact, then $\Xi_{K,S}$ is compact.
\label{prop6}
\end{proposition}
\begin{proof}
  Let $\{v_{n}\}$ be a sequence in $\Xi_{K,S}$. It can be shown that
  there exists a subsequence $\{v_{n_{k}}\}$ converging to $v$
  whose marginals belong to $K$ and $S$ \cite[Lemma~4.4]{Vil09}. Since each
  $v_{n_{k}}$ is concentrated on a $c$-cyclically monotone set by
  assumption, it can be shown by using the continuity of $c$ that
  $v$ is also concentrated on a $c$-cyclically monotone set (see proof of
  Theorem 5.20 in \cite{Vil09}). Then $v$ is also an element
  of $\Xi_{K,S}$ by \cite[Theorem~B]{Pra08}.
\end{proof}

Since $\{\mu\}$ and $\P_{M}(\sY)$ are both compact, we obtain that
$\Xi_{\{\mu\},\P_{M}(\sY)}$ is compact. Thus it follows that
$\P(\Xi_{\{\mu\},\P_{M}(\sY)})$ is also compact. Furthermore, by Proposition
\ref{prop7} we have $\Xi_{\{\mu\},\P_{M}(\sY)} \supset \{v\in
\hat{\Gamma}_{\mu}(M): L(v)<\infty\}$.  Hence the randomization can be
restricted to $\Xi_{\{\mu\},\P_{M}(\sY)}$ when defining
$\hat{\Gamma}_{\mu}^{\rm R}(M)$ for \textbf{(P2)}. Let $\Xi_{\{\mu\},\P_{M}(\sY)}^{\rm R}$
be the randomization of $\Xi_{\{\mu\},\P_{M}(\sY)}$ obtained by replacing
$\Gamma_{\mu}(M)$ with $\Xi_{\{\mu\},\P_{M}(\sY)}$ in (\ref{neweq4}).  One can
show that the mapping $\P(\Xi_{\{\mu\},\P_{M}(\sY)})\ni P\mapsto v_P \in
\Xi_{\{\mu\},\P_{M}(\sY)}^{\rm R}$ is continuous by using the same proof as in Lemma
\ref{lemma6}.  Thus $\Xi_{\{\mu\},\P_{M}(\sY)}^{\rm R}$ is the continuous image of a
compact set, and thus it is also compact.  This, together with the
compactness of $\Gamma_{\mu,\psi}$ and the lower semicontinuity of $L$, implies
the existence of the minimizer for \textbf{(P2)} under Assumption~1. 

To tie up a loose end, we still have to show that $\hat{\Gamma}_{\mu}(M)$ is
measurable, which will complete the proof under Assumption~1.

\begin{lemma}
\label{lem_meas}
 $\hat{\Gamma}_{\mu}(M)$ is a Borel set.
\end{lemma}

\begin{proof}
  Let us define $\hat{\Gamma}^{{\rm
      f}}_{\mu}(M):=\{v\in\hat{\Gamma}_{\mu}(M):L(v)<\infty\}$ and
  $\hat{\Gamma}^{\infty}_{\mu}(M)=\hat{\Gamma}_{\mu}(M)\setminus
  \hat{\Gamma}^{f}_{\mu}(M)$.  Since solutions to the optimal transport problem
  having finite costs must concentrate on $c$-cyclically monotone sets by
  Proposition~\ref{prop7}, we have $\hat{\Gamma}^{{\rm
      f}}_{\mu}(M)=\{v\in\Xi_{\{\mu\},\P_M(\sY)}:L(v)<\infty\}$. Hence,
  $\hat{\Gamma}^{f}_{\mu}(M)$ is a Borel set since $\Xi_{\{\mu\},\P_M(\sY)}$ is
  compact and $L$ is lower semi-continuous. Recall the continuous mapping $f_2$
  in the proof of Lemma~\ref{lemma4}. Since $\Xi_{\{\mu\},\P_M(\sY)}$ is compact,
  $\{v\in\Xi_{\{\mu\},\P_M(\sY)}:L(v)\leq N\}$ is also compact for all $N\ge
  0$. Hence, $f_2\bigl(\hat{\Gamma}^{{\rm
      f}}_{\mu}(M)\bigr)=\bigcup_{N=0}^{\infty}
  f_2\bigl(\{v\in\Xi_{\{\mu\},\P_M(\sY)}:L(v)\leq N\}\bigr)$ is $\sigma$-compact,
  so a Borel set, in $\P_M(\sY)$. Since
  $f_2\bigl(\hat{\Gamma}^{\infty}_{\mu}(M)\bigr)=\P_M(\sY)\setminus
  f_2\bigl(\hat{\Gamma}^{f}_{\mu}(M)\bigr)$,
  $f_2\bigl(\hat{\Gamma}^{\infty}_{\mu}(M)\bigr)$ is also a Borel set. Note that for
  any $v\in\hat{\Gamma}^{\infty}_{\mu}(M)$ we have $L(v)=\infty$, which means
  that all $\tilde{v}$ with the same marginals as  $v$ are also in $\hat{\Gamma}^{\infty}_{\mu}(M)$. This implies
  $\hat{\Gamma}^{\infty}_{\mu}(M)=f_2^{-1}\bigl(f_2\bigl(\hat{\Gamma}^{\infty}_{\mu}(M)\bigr)\bigr)$. Hence,
  $\hat{\Gamma}^{\infty}_{\mu}(M)$ is a Borel set.
\end{proof}

\hspace{-0.4cm}\textit{II) \textbf{Proof under Assumption~2}}

It is easy to check that the proof under Assumption~1 remains valid if $\sX$ and
$\sY$ are arbitrary uncountable Polish spaces such that $\sY$ is compact, and the
distortion measure $\rho$ is an extended real valued function (no steps
exploited the special structure of $\R^n$). Let $\sY$ be the one-point
compactification of $\R^{n}$ \cite{Dud89}. $\sY$ is clearly an uncountable
Polish space. Define the extended real valued distortion measure
$\rho:\sX\times\sY\rightarrow[0,\infty]$ by
\begin{align}
\rho(x,y)=\begin{cases}
\|x-y\|^{2},   &\text{ if } y\in \R^{n}  \\
\infty,  &\text{ if } y=\infty.
\end{cases}
\label{neweq5}
\end{align}
It is straightforward to check that $\rho$ is continuous. 
Define $L$ on $\P(\sX\times\sY)$ as before, but with this new distortion measure
$\rho$. The proof under Assumption~1 gives a minimizer
$v^{*}=\int_{\Gamma_{\mu}(M)} v P(dv)$ for $\textbf{(P1)}$. Define
$\tilde{\Gamma}_{\mu}(M)=\{v\in\Gamma_{\mu}(M):
v(\sX\times\{\infty\})=0\}$. Since $L(v^{*})<\infty$ by assumption,
$P(\tilde{\Gamma}_{\mu}(M))=1$. This implies that $v^{*}$ is also a minimizer
for the problem $\textbf{(P1)}$ when $\sX=\sY=\R^{n}$ and
$\rho=\|x-y\|^{2}$.\qed

\subsection{\textbf{Proof of Theorem \ref{thconv}}}
\label{app3a}

From the proof of Theorem~\ref{thmopt} recall the set
$\hat{\Gamma}_{\mu}(M)$ of probability measures which solve the
optimal mass transport problem for fixed input marginal $\mu$ and some
output marginal $\psi_0$ in $\mathcal{P}_{M}(\sY)$. It is known that if $\mu$ admits a density and
$\rho(x,y)=\|x-y\|^2$, then each $v\in \hat{\Gamma}_{\mu}(M)$ is in
the form $v(dx\,dy)=\mu(dx)\delta_{q(x)}(dy)$ for some $q\in \Q_{M,{\rm c}}$
(see, e.g. \cite[Theorem~1]{McMo99}). Thus in this case $
\hat{\Gamma}_{\mu}(M)\subset \Gamma_{\mu}(M)$, which implies that
$\hat{\Gamma}_{\mu,\psi}^{\rm R}(M)\subset
\Gamma_{\mu,\psi}^{{\rm {R,c}}}(M)\subset
\Gamma_{\mu,\psi}^{\rm R}(M) $. Recall the problem $\textbf{(P2)}$ in the proof of Theorem \ref{thmopt}. It was shown that $\textbf{(P2)}$ has a minimizer $v^{*}$. It is clear from the previous discussion that  $v^{*}$ is obtained by randomizing over the set of quantizers having
convex codecells represented by $\hat{\Gamma}_{\mu}(M)$. On the other hand, $v^{*}$ is also a minimizer  for the problem $\textbf{(P1)}$ by Proposition \ref{prop5} in the proof of Theorem \ref{thmopt}.
\qed

\subsection{\textbf{Proof of Theorem~\ref{thm_finite} }}
\label{app4}

Recall the continuous mapping $h: \P(\Gamma_{\mu}(M)) \to
\Gamma_{\mu}^{\rm R}(M)$ defined in \eqref{eq_hmap}. Let
$\P_{F}(\Gamma_{\mu}(M))$ denote the set of probability measures on
$\Gamma_{\mu}(M)$ having finite support. Clearly
$h(\P_{F}(\Gamma_{\mu}(M)))=\Gamma_{\mu}^{{\rm FR}}(M)$.

\begin{lemma}
$\Gamma_{\mu}^{{\rm FR}}(M)$ is dense in $\Gamma_{\mu}^{\rm R}(M)$.
\label{lemma7}
\end{lemma}
\begin{proof}
  Since $\Gamma_{\mu}(M)$ is a separable metric space, $\P_{F}(\Gamma_{\mu}(M))$
  is dense in $\P(\Gamma_{\mu}(M))$ by \cite[Theorem 6.3]{Par67}. Since
  $\Gamma_{\mu}^{{\rm FR}}(M)$ is the image of a $\P_{F}(\Gamma_{\mu}(M))$ under the
  continuous function $h$ which maps $\P(\Gamma_{\mu}(M))$ onto 
  $\Gamma_{\mu}^{\rm R}(M)$, it is dense in $\Gamma_{\mu}^{\rm R}(M)$.
\end{proof}

Recall that the Prokhorov metric on $\P(\sE)$, where $(\sE,d)$ is a metric
space, is defined as \cite{Bil99}
\begin{eqnarray}
\label{eq:prokh}
  \!\!\!\!\!   \!\!\!\!d_{P}(v,\nu) \!\!\!\!
&= & \!\!\!\!\inf\big\{\alpha : v(A)\leq \nu(A^{\alpha})+\alpha, \nonumber  \\
  &&\quad   \nu(A)\leq
v(A^{\alpha})+\alpha \text{ for all }A \in \B(\sE)\big \} 
\end{eqnarray}
where 
\[
A^{\alpha}=\Big\{e\in \sE: \inf_{e'\in A}
d(e,e')<\alpha\Big\}.
\]
Hence for $v,\nu\in \P(\sX\times \sY)$, 
\begin{eqnarray*}
  d_{P}(v,\nu) &\geq & \inf\big\{\alpha: v(\sX\times B)\leq \nu((\sX\times
  B)^{\alpha})+\alpha,  \\
& &   \qquad \nu(\sX\times B)\leq v((\sX\times B)^{\alpha})+\alpha, B \in
\B(\sY) \big\}  \\ 
&=&d_{P}\big(v(\sX\times\,\cdot\,),\nu(\sX\times\,\cdot\,)\big) 
\end{eqnarray*}
(note that  $(\sX\times B)^{\alpha}=\sX\times B^{\alpha}$). This implies
\begin{eqnarray}
G^{\alpha}_{\psi} & := & \{v \in \P(\sX\times \sY): v(\sX\times\,\cdot\,)\in
B(\psi,\alpha)\} \nonumber  \\
&\supset & \{v \in \P(\sX\times \sY): d_{P}(\hat{v},v)<\alpha\},
\label{neweq6}
\end{eqnarray}
where $\hat{v}$ is such that $\hat{v}(\sX\times\,\cdot\,)=\psi$ and $\alpha>0$ .
Recall that given a metric space $\sE$ and $A\subset \sE$, a  set $B\subset A$
is relatively open in $A$ if  $B=A\cap U$ for some open set
$U\subset \sE$.  

\begin{lemma}
$\M^{\delta}_{\mu,\psi}$ is relatively open in $\Gamma_{\mu}^{\rm R}(M)$.
\label{lemma9}
\end{lemma}
\begin{proof}
Since $\M^{\delta}_{\mu,\psi}=G^{\delta}_{\psi}\cap \Gamma_{\mu}^{\rm R}(M)$, it is
enough to prove that  $G^{\delta}_{\psi}$ is open in $\P(\sX\times\sY)$. Let
$\tilde{v}\in G^{\delta}_{\psi}$. Then $\tilde{v}(\sX\times\,\cdot\,)\in
B(\psi,\delta)$ by definition, and there exists $\delta_{0}> 0$ such that $B(\tilde{v}(\sX\times\,\cdot\,),\delta_{0}) \subset B(\psi,\delta)$. By (\ref{neweq6}) we have
\begin{align}
\big\{v\in\P(\sX\times\sY): d_{P}(\tilde{v},v)<\delta_{0}\big\}\subset
G^{\delta_{0}}_{v(\sX\times\,\cdot\,)} \, . 
\end{align}
We also have $G^{\delta_{0}}_{v(\sX\times\,\cdot\,)}\subset G^{\delta}_{\psi}$ since $B(\tilde{v}(\sX\times\,\cdot\,),\delta_{0}) \subset B(\psi,\delta)$. This implies that $G^{\delta}_{\psi}$ is open in $\P(\sX\times\sY)$.
\end{proof}

\hspace{-0.4cm}\textit{I) \textbf{Case 1}}

First we treat the case $L(v)>\inf_{v'\in \Gamma_{\mu}(M)} L(v')$. If $\rho$ is continuous
and bounded, then $L$ is continuous.  Hence, $\{v'\in\Gamma_{\mu}^{\rm R}(M):
L(v')<L(v)\}$ is relatively open in $\Gamma_{\mu}^{\rm R}(M)$. Define
$F:=\{v'\in\Gamma_{\mu}^{\rm R}(M): L(v')<L(v)\}$. 

\begin{lemma}
$F\cap \M_{\mu,\psi}^{\delta}$ is \emph{nonempty} and relatively open in $\Gamma_{\mu}^{\rm R}(M)$.
\label{lemma10}
\end{lemma}
\begin{proof}
  By Lemma \ref{lemma9} and the above discussion the intersection is clearly
  relatively open in $\Gamma_{\mu}^{\rm R}(M)$, so we need to show that it is not
  empty.  Since $L(v)>\inf_{v'\in \Gamma_{\mu}(M)} L(v')$, there exists
  $\tilde{v}\in\Gamma_{\mu}(M)$ such that $L(\tilde{v})<L(v)$. Define the
  sequence of randomized quantizers $\{v_{n}\}\in\Gamma_{\mu}^{\rm R}(M)$ by letting
  $v_{n}=\frac{1}{n}\tilde{v}+(1-\frac{1}{n})v$. Then, $v_{n}\rightarrow v$
  weakly because for any  continuous and bounded real function $f$ on $\sX\times
  \sY$  
\begin{eqnarray*}
\lefteqn{ \lim_{n\rightarrow \infty}\left|\int_{\sX\times \sY}f dv_{n}-\int_{\sX\times
    \sY}f dv\right|  }\\ 
&=&\lim_{n\rightarrow \infty}\frac{1}{n}\left|\int_{\sX\times \sY}f
  d\tilde{v}-\int_{\sX\times \sY}f dv\right| =0. 
\end{eqnarray*}
Hence there exists $n_{0}$ such that $v_{n}\in M^{\delta}_{\mu,\psi}$ for all $n\geq n_{0}$. On the other hand, for any $n$
\begin{align}
L(v_{n})&=L\left(\frac{1}{n}\tilde{v}+\Big(1-\frac{1}{n}\Big)v\right) \nonumber \\
&=\frac{1}{n}L(\tilde{v})+\Big(1-\frac{1}{n}\Big)L(v)\nonumber \\
&<L(v) \nonumber.
\end{align}
This implies $v_{n}\in \M^{\delta}_{\mu,\psi}\cap F$ for all $n\geq n_{0}$, completing the proof.
\end{proof}
Hence, we can conclude that there exists finitely randomized quantizer $v_{F}\in
F\cap M_{\mu,\psi}^{\delta}$ by Lemmas \ref{lemma7} and \ref{lemma10}. By the
definition of $F$ we also have $L(v_{F})<L(v)$. This completes the proof of the
theorem  for this case.

\hspace{-0.4cm}\textit{II) \textbf{Case 2}}

The case   $L(v)=\inf_{v'\in \Gamma_{\mu}(M)} L(v'):=L^*$ is handled
similarly. Define the subset of $\Gamma_{\mu}(M)$ whose elements correspond to optimal quantizers:
\begin{align}
\Gamma_{\mu,{\rm opt}}(M)=\{v^{'}\in\Gamma_{\mu}(M): L(v^{'})=L^{*}\}\nonumber.
\end{align}
Define $\Gamma_{\mu,\rm{ opt}}(M)=L^{-1}(L^{*})\cap\Gamma_{\mu}(M)$ and let
$\Gamma_{\mu,\rm{ opt}}^{\rm R}(M)$ be the randomization of $\Gamma_{\mu,\rm{ opt}}(M)$, obtained
by replacing $\Gamma_{\mu}(M)$ with $\Gamma_{\mu,{\rm opt}}(M)$ in
(\ref{neweq4}). Note that if $L(v)=L^{*}$, then $v$ is obtained by randomizing
over the set $\Gamma_{\mu,{\rm opt}}(M)$, i.e., $v\in\Gamma_{\mu,{\rm opt}}^{\rm R}(M)$. Let
$\Gamma_{\mu,{\rm opt}}^{{\rm FR}}(M)$ denote the set obtained by the finite randomization
of $\Gamma_{\mu,{\rm opt}}(M)$. By using the same proof method as in Lemma
\ref{lemma7} we can prove that $\Gamma_{\mu,{\rm opt}}^{{\rm FR}}(M)$ is dense in
$\Gamma_{\mu,{\rm opt}}^{\rm R}(M)$. In addition, $\M_{\mu,\psi}^{\delta}$ is relatively
open in $\Gamma_{\mu,{\rm opt}}^{\rm R}(M)$ by Lemma \ref{lemma9}. Thus, there exists
finitely randomized quantizer $v_{F}\in
\M_{\mu,\psi}^{\delta}\cap\Gamma_{\mu,{\rm opt}}^{\rm R}(M)$ with
$L(v_{F})=L(v)=L^{*}$. This completes the proof of Theorem~\ref{thm_finite}.\qed

\subsection{\textbf{Proof of Theorem~\ref{thm_finarb}}}
\label{app4a}

Let $\hat{v}\in \M_{\mu,\psi}^{\delta}$ be such that $L(\hat{v})<\inf_{v\in
  \M^{\delta}_{\mu,\psi}}L(v)+\varepsilon/2$. Let $\hat{P}$ be the probability
measure on $\Gamma_{\mu}(M)$ that induces $\hat{v}$,
i.e., $\hat{v}=\int_{\Gamma_{\mu}(M)} v \hat{P}(dv)$. Consider a sequence of
independent and identically distributed (i.i.d$.$) random variables
$X_1,X_1,\ldots,X_n,\ldots$ defined on some probability space
$(\Omega,\mathcal{F},\gamma)$ which take values in
$\bigl(\Gamma_{\mu}(M),\B(\Gamma_{\mu}(M)) \bigr)$ and have common distribution
$\hat{P}$. Then $L(X_{1}), L(X_{2}), \ldots$ are i.i.d.\ $\R$-valued random
variables with distribution $\hat{P}\circ L^{-1}$. Thus we have
\begin{align*}
\int_{\Omega}
L(X_{i}(\omega))\gamma(d\omega)&=\int_{\Gamma_{\mu}(M)}L(v)\hat{P}(dv)=L(\hat{v})\\
& <\inf_{v\in M^{\delta}_{\mu,\psi}}L(v)+\frac{\varepsilon}{2} \nonumber
\end{align*}
by assumption. The empirical measures $P_{n}^{\omega}$ on $\Gamma_{\mu}(M)$
corresponding to $X_1,\ldots,X_n$ are  
\begin{align}
P_{n}^{\omega}(\,\cdot\,):=\frac{1}{n}\sum_{i=1}^{n}\delta_{X_{i}(\omega)}(\,\cdot\,). \nonumber
\end{align}
By the strong law of large numbers
\begin{align}
\frac{1}{n}\sum_{i=1}^{n}L(X_{i})&=\int_{\Gamma_{\mu}(M)}L(v)P_{n}^{\omega}(dv)
  \nonumber \\*
 &\rightarrow\int_{\Gamma_{\mu}(M)}L(v)\hat{P}(dv)=L(\hat{v}) \label{eq6}
\end{align}
$\gamma$-a.s.  As a subset of $\P(\sX\times \sY)$, $\Gamma_{\mu}(M)$ with the
Prokhorov metric is a separable metric space, and thus by \cite[Theorem
11.4.1]{Dud89} we also have the almost sure convergence of empirical measures,
i.e., $P_{n}^{\omega}\rightarrow\hat{P}$ weakly $\gamma$-a.s. Thus there exists
$\hat{\omega}\in \Omega$ for which both convergence results hold. Define the
sequence of finitely randomized quantizers $\{v_{n}\}$ by
$v_{n}=\int_{\Gamma_{\mu}(M)}v P_{n}^{\hat{\omega}}(dv)$. By (\ref{eq6})
$L(v_{n})\rightarrow L(\hat{v})$ and by Lemma \ref{lemma6} in the proof of
Theorem \ref{thmopt} $v_{n}\rightarrow\hat{v}$ weakly. Since
$\M_{\mu,\psi}^{\delta}$ is a relatively open neighborhood of $\hat{v}$ in
$\Gamma_{\mu}^{\rm R}(M)$, we can find sufficiently large $n$ such that
$v_{n}\in \M_{\mu,\psi}^{\delta}$ and
$L(v_{n})<L(\hat{v})+\frac{\varepsilon}{2}$. Hence, for any $\varepsilon>0$
there exists an $\varepsilon$-optimal finitely randomized quantizer for
\textbf{(P3)}.  \qed

\subsection{\textbf{Proofs for Section~\ref{sec5}}}
\label{app5}

\begin{proof}[Proof of Lemma~\ref{alemma3}]
The proof uses standard notation for information quantities \cite{CoTh06}. 
  Let $X^n \sim \mu^n$, $Z \sim \nu$, and $Y^n=\sq(X^n,Z) \sim \psi^n$, where
  $(\sq,\nu)$ is an arbitrary Model~2 randomized quantizer with at most $2^{nR}$
  levels ($Z$ is independent of $X^n$). Let $D_i=E[\rho(X_i,Y_i)]$ and
  $D=\frac{1}{n}\sum_{i=1}^n D_i=E[\rho_n(X^n,Y^n)]$. Since $\sq(\,\cdot\,,z)$ has
  at most $2^{nR}$ levels for each $z$,
  \begin{align}
    nR \geq H(Y^n|Z) &\geq  I(X^n;Y^n|Z) \nonumber \\
    &\geq  I(X^n;Y^n) \label{aeq13} \\
    &\geq \sum_{i=1}^{n} I(X_i;Y_i) \label{aeq13a}  \\
    &\geq  \sum_{i=1}^{n} I_m(\mu\|\psi,D_i)  \nonumber \\
    &\geq n I_m(\mu\|\psi,D)  \nonumber
\end{align}
where in the last two inequalities follow since $Y_i\sim \psi$, $i=1,\ldots,n$
and $I_m(\mu\|\psi,D)$ is convex in $D$ \cite[Appendix A]{ZaRo01}. Inequalities
\eqref{aeq13} and \eqref{aeq13a} follow from the chain rule for mutual
information (Kolmogorov's formula) \cite[Corollary 7.14]{Gra11}, which in
particular implies that $I(U;V|W)\ge I(U;V)$ for general random variables $U$, $V$,
and $W$, defined on the same probability space, such that $U$ and $W$ are
independent.  This proves that  $R\ge I_{m}(\mu
\| \psi, D)$.
\end{proof}

\begin{proof}[Proof of Lemma~\ref{alemma1}]
Let $U^{2^{nR}}=\bigl(U^{n}(1),\ldots,U^{n}(2^{nR})\bigr)$ which is a $n2^{nR}$-vector. Then, we can write
\begin{align}
\hat{X}^{n}=g(X^{n},U^{2^{nR}}) \nonumber
\end{align}
for a function $g$ from $\sY^{n(2^{nR}+1)}$ to $\sY^{n}$. Observe the following:
\begin{itemize}
\item[\emph{(i)}] For any permutation $\sigma$ of $\{1,\ldots,n\}$, $X^{n}$ and
  $X_{\sigma}^{n}=\bigl(X_{\sigma(1)},\ldots,X_{\sigma(n)}\bigr)$ have the same
  distribution. The same issue is true for $U^{n}(i)$ and $U^{n}(i)_{\sigma}$
  for all $i$ because for any $u^n \in T_n(\psi_n)$, $u^{n}_{\sigma} \in
  T_n(\psi_n)$ and this mapping is a bijection on $T_n(\psi_n)$. It follows from
  the independence of $X^{n}$ and $U^{n}(i)$ that $(X^{n}, U^{nR})$ and
  $(X_{\sigma}^{n},U_{\sigma}^{2^{nR}})$ have the same distribution, where
  $U_{\sigma}^{2^{nR}}\coloneqq\bigl(U^{n}(1)_{\sigma},\ldots,U^{n}(2^{nR})_{\sigma}\bigr)$. Thus,
  $g(X^{n},U^{2^{nR}})$ and $g(X^{n}_{\sigma},U^{2^{nR}}_{\sigma})$ have
  the same distribution.
\item[\emph{(ii)}] For any $x^{n}\in\sX^{n}$ and $y^{n}\in\sY^{n}$, $\rho_n(x^{n},y^{n})=\rho_n(x_{\sigma}^{n},y_{\sigma}^{n})$. Thus, if $g$ outputs $u^{n}(i)$ for inputs $x^{n},u^{n}(1),\ldots,u^{n}(2^{nR})$, then $g$ outputs $u^{n}(i)_{\sigma}$ for inputs  $x^{n}_{\sigma},u^{n}(1)_{\sigma},\ldots,u^{n}(2^{nR})_{\sigma}$. It follows that
    \begin{align}
    g(X_{\sigma}^{n},U_{\sigma}^{2^{nR}})=g(X^{n},U^{2^{nR}})_{\sigma} \nonumber.
    \end{align}
    Together with $i)$ this implies that $\hat{X}^{n}$ and $\hat{X}^{n}_{\sigma}$ have the same distribution.
\end{itemize}
Let $u^{n}$ and $v^{n}\in T_{n}(\psi^{(n)}_{n})$ and so $u^{n}=v^{n}_{\sigma}$ for some permutation $\sigma$. Then \emph{(ii)} implies
\begin{align}
\Pr\{\hat{X}^{n}=u^{n}\}=\Pr\{\hat{X}^{n}_{\sigma}=u^{n}\}. \nonumber
\end{align}
Since $\Pr\{\hat{X}^{n}=v^{n}\}=\Pr\{\hat{X}^{n}_{\sigma}=v^{n}_{\sigma}\}$ and
$v^{n}_{\sigma}=u^{n}$, we obtain 
\begin{align}
\Pr\{\hat{X}^{n}=u^{n}\}=\Pr\{\hat{X}^{n}=v^{n}\} \nonumber
\end{align}
proving that $\hat{X}^{n}$ is uniform on $T_{n}(\psi^{(n)}_{n})$.
\end{proof}

\begin{proof}[Proof of Lemma~\ref{alemma2}] 
By \cite[Theorem 11.1.2]{CoTh06} we have
\begin{align}
  \frac{1}{n} \mathcal{D}(\psi^{(n)}_{n}\|\psi^{n})&=\frac{1}{n} \sum_{y^{n}\in T_{n}(\psi_{n})} \psi^{(n)}_{n}(y^{n}) \log\frac{\psi^{(n)}_{n}(y^{n})}{\psi^{n}(y^{n})} \nonumber \\
  &=\frac{1}{n} \log
  \frac{2^{n(H(\psi_{n})+\mathcal{D}(\psi_{n}\|\psi))}}{|T_{n}(\psi_{n})|}.
\end{align}
From  \cite[Theorem 11.1.3]{CoTh06},
\[
\frac{1}{(n+1)^{|\sX|}}2^{nH(\psi_{n})} \le |T_{n}(\psi_{n})| \le 2^{nH(\psi_{n})}
\]
and thus $  \frac{1}{n} \mathcal{D}(\psi^{(n)}_{n}\|\psi^{n}) $ is sandwiched 
between $\mathcal{D}(\psi_{n}\|\psi)$ and $\frac{|\sX|}{n}\log(n+1) +
\mathcal{D}(\psi_{n}\|\psi)$. Thus 
\[
\lim_{n\to \infty} \frac{1}{n} \mathcal{D}(\psi^{(n)}_{n}\|\psi^{n}) =
\lim_{n\to \infty} \mathcal{D}(\psi_{n}\|\psi) =0
\]
where the second limit holds since $\sX$ is a finite set and 
$\psi_n\rightarrow\psi$ in the $l_1$-distance. 
\end{proof}

\begin{proof}[Proof of Lemma~\ref{lem_coupling}]

\noindent Let $\rho^{H}$ denote the Hamming distortion and let
$\rho^H_n(x^n,y^n) = (1/n) \sum_{i=1}^n \rho^H(x_i,y_i)$. Since $\rho(x,x)=0$
for all $x\in \sX$, we  have 
\[
\rho_n(x^n,y^n) \leq \rho_{\rm max} \, \rho^H_n(x^n,y^n). 
\]
Let $T^H_n(\psi^{(n)}_{n},\psi^{n})$ be the distortion of the optimal coupling
between $\psi^{(n)}_{n}$ and $\psi^{n}$ when the cost function is
$\rho^H_n$. Then the above inequality gives
\[
\hat{T}_n(\psi^{(n)}_{n},\psi^{n})\leq  \rho_{\rm max} \,T^H_n(\psi^{(n)}_{n},\psi^{n}).
\]
On the other hand,  by  Marton's inequality \cite[Proposition~1]{Mar96}
\[
T^H_n(\psi^{(n)}_{n},\psi^{n})\leq
  \sqrt{\frac{1}{2n}D(\psi^{(n)}_{n}\|\psi^{n})} .
\]
Combining these bounds with $ \frac{1}{n}D(\psi^{(n)}_{n}\|\psi^{n}) \rightarrow
0$ (Lemma~\ref{alemma2}), we obtain
\begin{align}
\label{eq_tr}
\lim_{n\to \infty} \hat{T}_n(\psi^{(n)}_{n},\psi^{n}) =0
\end{align}
which is the first statement of the lemma.

Recall that $\rho(x,y)=d(x,y)^p$ for some $p>0$, where $d$ is a metric.  Let
$q=\max\{1,p\}$.  If $p\ge 1$, then   $\|V^n\|_p \coloneqq \big( E\big[\, \sum_{i=1}^n
|V_i|^p\, \big] \big)^{1/q}$ is a norm on $\R^n$-valued random vectors whose
components have finite $p$th moments, and if $1<p<0$, we still have 
$\|U^n+V^n\|_p\le \|U^n\|_p+\|V^n\|_p$. Thus we can upper bound $
E[\rho_n(X^{n},Y^{n})]$ as follows:
\begin{eqnarray*}
\lefteqn{ \biggl( E\biggl[ \frac{1}{n} \sum_{i=1}^n \rho(X_i,Y_i)   \biggr]
  \biggr)^{1/q}  } \\
   &=& \biggl( E\biggl[ \frac{1}{n} \sum_{i=1}^n d(X_i,Y_i)^p   \biggr] \biggr)^{1/q}
  \\*
&\le &  \biggl( E\biggl[ \frac{1}{n} \sum_{i=1}^n d(X_i,\hat{X}_i)^p   \biggr]
\biggr)^{1/q} \!\!\! + \biggl( E\biggl[ \frac{1}{n} \sum_{i=1}^n d(\hat{X}_i,Y_i)^p   \biggr] \biggr)^{1/q}
  \\
  &=  &  \Bigl(E[\rho_n(X^{n},\hat{X}^{n})]\Bigr)^{1/q} +
  \hat{T}_n(\psi^{(n)}_{n},\psi^{n})^{1/q}.
\end{eqnarray*}
Hence  \eqref{aeq4} and \eqref{eq_tr} imply 
\[
\limsup_{n\to \infty}  E[\rho_n(X^{n},Y^{n})] \le D
\]
as claimed. \end{proof}

\begin{proof}[Proof of Lemma~\ref{alemma4}] Let $X \sim \mu$ and $Y \sim
  \psi$ such that $I(X;Y)$ achieves $I_m(\mu\|\psi,D)<\infty$ at distortion
  level $D$ (the existence of such pair follows from an analogous statements for
  rate-distortion functions \cite{Csi74}) . Let $q_k$ denote the uniform
  quantizer on the interval $[-k,k]$ having $2^k$ levels, where we extend $q_k$
  to the real line by using the nearest neighborhood encoding rule. Let $X(k)
  =q_k(X)$ and
  $Y(k)=q_k(Y)$.  We clearly have
\begin{align}
E[(X-X(k))^2]\to 0,\;  E[(Y-Y(k))^2]\to 0 \text{ as $k\to \infty$}. 
\label{aeq14}
\end{align}
Let $\mu_k$ and $\psi_k$ denote the distributions of $X(k)$ and $Y(k)$,
respectively. Then by \cite[Theorem 6.9]{Vil09} it follows that
$\hat{T}_1(\mu_k,\mu)\to 0$ and $\hat{T}_1(\psi_k,\psi) \to 0$ as $k\rightarrow\infty$ since
$\mu_k\to \mu$, $\psi_k\to \psi$ weakly,  and
$E[X(k)^2]\to E[X^2]$,  $E[Y(k)^2]\to E[Y^2]$.

By the data processing inequality, we have for all $k$, 
\begin{align}
  I(X(k);Y(k)) \leq I(X;Y).
\label{aeq15}
\end{align}
Also note that  (\ref{aeq14}) implies 
\begin{eqnarray*}
\lefteqn{  \limsup_{k\to \infty} E\bigl[ \rho_1(X(k),Y(k))\bigr] }\qquad \\
& = & \limsup_{k\to \infty}
E\bigl[ \bigl( X(k)-Y(k)\bigr)^2 \bigr]  \le D.
\end{eqnarray*}
Thus, for given  $\varepsilon>0$,  if  $k$ is large we have
$I_m(\mu_k\|\psi_k,D+\varepsilon) \leq I_m(\mu\|\psi,D)$ as claimed. 
\end{proof}

\section*{Acknowledgement} The authors would like to  thank Ram Zamir for helpful
discussions concerning  the proof of Theorem~\ref{thm_rd}.  Naci Saldi would
like to thank Marcos Vasconcelos for discussions on optimal transport theory. 

\bibliographystyle{IEEEtran}


\end{document}